%% file: uploaded_arxiv.tex
\documentclass[onecolumn, a4paper, 12pt, conference]{ieeeconf}      % Use this line for a4 paper
\IEEEoverridecommandlockouts                              % This command is only needed if 
\let\labelindent\relax

\usepackage[utf8]{inputenc}
\usepackage{graphics} % for pdf, bitmapped graphics files
\usepackage{amsmath} % assumes amsmath package installed
\usepackage{amssymb}  % assumes amsmath package installed
\usepackage{amsfonts}
\usepackage{comment}
\usepackage{enumitem}
\usepackage{cases,xspace,enumitem}
\usepackage{version}
\includeversion{tac}
\excludeversion{arxiv}
\usepackage{subcaption}
\usepackage{multirow}
\usepackage[rightcaption]{sidecap}

\DeclareCaptionLabelFormat{figure}{{Figure #2}}
\usepackage{tikz}
\usetikzlibrary{positioning}

\usepackage[prependcaption,colorinlistoftodos]{todonotes}
\definecolor{gnred}{RGB}{255,91,89}

\definecolor{gnred1}{RGB}{71,0,0} % 470000
\definecolor{gnred2}{RGB}{117,0,0} % 750000
\definecolor{gnred3}{RGB}{164,0,0} % a40000
\definecolor{gnred4}{RGB}{211,0,0} % d30000
\definecolor{gnred5}{RGB}{255,0,0} % FF0000
\definecolor{gnred6}{RGB}{255,42,34} % FF2a22
\definecolor{gnred7}{RGB}{255,91,89} % ff5b59 --- favorite

\definecolor{gnblue1}{RGB}{0,36,71}   % 002447
\definecolor{gnblue2}{RGB}{0,60,118}  % 003c76
\definecolor{gnblue3}{RGB}{0,85,164}
\definecolor{gnblue4}{RGB}{0,108,212}
\definecolor{gnblue4}{RGB}{0,108,212}
\definecolor{gnblue5}{RGB}{0,133,255}  % 0085ff
\definecolor{gnblue6}{RGB}{35,156,255} % 239cff
\definecolor{gnblue7}{RGB}{88,177,255} % 58b1ff

\definecolor{gnbrown1}{RGB}{71,27,0}  % 471b00
\definecolor{gnbrown2}{RGB}{117,45,0} % 752d00
\definecolor{gnbrown3}{RGB}{164,62,0} % a43e00
\definecolor{gnbrown4}{RGB}{211,80,0} % d35000
\definecolor{gnbrown5}{RGB}{255,97,0} % ff6100
\definecolor{gnbrown6}{RGB}{255,127,26} % ff7f1a
\definecolor{gnbrown7}{RGB}{255,155,86} % ff9b56

\renewcommand{\theenumi}{(\roman{enumi})}

\newcommand\Item[1][]{%
  \ifx\relax#1\relax  \item \else \item[#1] \fi
  \abovedisplayskip=0pt\abovedisplayshortskip=0pt~\vspace*{-\baselineskip}}

\usepackage{hyperref}
\hypersetup{
    colorlinks=true,
    linkcolor=blue,
    filecolor=magenta,      
    urlcolor=cyan,
    pdfpagemode=FullScreen,
    }

\usepackage{amsthm}

\newtheoremstyle{ieeeconf}
  {0pt}   % ABOVESPACE
  {0pt}   % BELOWSPACE
  {\normalfont}  % BODYFONT
  {\parindent}       % INDENT (empty value is the same as 0pt)
  {\itshape} % HEADFONT
  {:}         % HEADPUNCT
  { } % HEADSPACE
  {\thmname{#1} \thmnumber{#2}\thmnote{ (#3)}} % CUSTOM-HEAD-SPEC
\makeatletter
\renewenvironment{proof}[1][\proofname]{\par
  \pushQED{\qed}%
  \normalfont \topsep\z@
  \trivlist
  \item[\hskip2em
        \itshape
    #1\@addpunct{:}]\ignorespaces
}{%
  \popQED\endtrivlist\@endpefalse
}
\makeatletter

\theoremstyle{ieeeconf}
\input{my_commands.tex}

\newcommand\oprocendsymbol{\hbox{$\triangle$}}
\newcommand\oprocend{\relax\ifmmode\else\unskip\hfill\fi\oprocendsymbol}

\usepackage{titlesec}

\titlespacing*{\section}
{0pt}{3ex plus 1ex minus .2ex}{2ex plus .2ex}
\titlespacing*{\subsection}
{0pt}{3ex plus 1ex minus .2ex}{2ex plus .2ex}
\titlespacing*{\subsubsection}
{0pt}{1.5ex plus 1ex minus .2ex}{1ex plus .2ex}
\titlespacing*{\paragraph}
{0pt}{1.5ex plus 1ex minus .2ex}{1ex plus .2ex}
\title{Positive Competitive Networks for Sparse Reconstruction}
\author{Veronica Centorrino$^a$, Anand Gokhale$^b$, Alexander Davydov$^b$,\\
Giovanni Russo$^c$, and Francesco Bullo$^b$
\thanks{$^a$Veronica Centorrino is with Scuola Superiore Meridionale, Italy. {\tt\small veronica.centorrino@unina.it.}}
\thanks{$^b$Anand Gokhale, Alexander Davydov, and Francesco Bullo are with the Center for Control, Dynamical 
Systems, and Computation, UC Santa Barbara, Santa Barbara, CA 93106 USA. {\tt\small 
anand\_gokhale@ucsb.edu, davydov@ucsb.edu, bullo@ucsb.edu}.}
\thanks{$^c$Giovanni Russo is with the Department of Information and Electric Engineering and Applied Mathematics, University of Salerno, Italy. {\tt\small giovarusso@unisa.it.}}
}

\date{}
\begin{document}
\pagestyle{plain} % Add this line to include page numbers

\maketitle

\begin{abstract}
\normalsize
We propose and analyze a continuous-time firing-rate neural network, the positive firing-rate competitive network (\pfcn), to tackle sparse reconstruction problems with non-negativity constraints. These problems, which involve approximating a given input stimulus from a dictionary using a set of sparse (active) neurons, play a key role in a wide range of domains, including for example neuroscience, signal processing, and machine learning. First, by leveraging the theory of proximal operators, we relate the equilibria of a family of continuous-time firing-rate neural networks to the optimal solutions of sparse reconstruction problems. Then, we prove that the \pfcn is a positive system and give rigorous conditions for the convergence to the equilibrium. Specifically, we show that the convergence: (i) only depends on a property of the dictionary; (ii) is linear-exponential, in the sense that initially the convergence rate is at worst linear and then, after a transient, it becomes exponential. We also prove a number of technical results to assess the contractivity properties of the neural dynamics of interest. Our analysis leverages contraction theory to characterize the behavior of a family of firing-rate competitive networks for sparse reconstruction with and without non-negativity constraints. Finally, we validate the effectiveness of our approach via a numerical example.
%%%%%%%%%%%%%%%%%%%%%%%%%%%%%%%%%
\end{abstract}

\section{Introduction}
Sparse reconstruction (SR) or sparse approximation problems are ubiquitous in a wide range of domains spanning, e.g., neuroscience, signal processing, compressed sensing, and machine learning~\cite{EJC-MBW:08, JW-AYY-AG-SSS-YM:08, ME-MATF-YM:10, JW-YM:22}. These problems involve approximating a given input stimulus from a dictionary, using a set of sparse (active) units/neurons. Over the past years, an increasing body of theoretical and experimental evidence~\cite{DHH-TNW:68, HBB:72, DJF:87, BAO-DJF:96, BAO-DJF:97, BAO-DJF:04} has grown to support the use of sparse representations in neural systems. In this context, we propose (and characterize the behavior of) a novel family of continuous-time, firing-rate neural networks (FNNs) that we show tackle SR problems. Due to their biological relevance, we are particularly interested in SR problem with non-negativity constraints and, to solve these problems, we propose the \emph{positive firing-rate competitive network}. This is an FNN whose state variables have the desirable, biologically plausible, property of remaining non-negative.

Historically, understanding representation in neural systems has been a key research challenge in neuroscience. The evidence that many sensory neural systems employ SR traces back to the pioneering work by Hubel and Wiesel, where it is shown that the responses of simple-cells in the mammalian visual cortex (V1) can be described as a linear filtering of the visual input~\cite{DHH-TNW:68}.
This insight was further expanded upon by Barlow, who hypothesized that sensory neurons aim to encode an accurate representation of the external world using the fewest active neurons possible~\cite{HBB:72}.
Subsequently, Field showed that simple-cells in V1 efficiently encode natural images using only a sparse fraction of active units~\cite{DJF:87}.
Then, Olshausen and Field proposed that biological vision systems encode sensory input data and showed that a neural network trained to reconstruct natural images with sparse activity constraints develops units with properties similar to those found in V1~\cite{BAO-DJF:96, BAO-DJF:97}. These ideas have since gained substantial support from studies on different animal species and the human brain~\cite{BAO-DJF:04}.

Formally, the SR problem can be formulated as a composite minimization problem\footnote{A composite minimization problem refers to an optimization task that involves minimizing a function composed of the sum of a differentiable and a non-differentiable component, typically combining a smooth loss function with a non-smooth regularization term.} given by a least squares optimization problem regularized with a sparsity-inducing penalty function.
While traditional optimization methods rely on discrete algorithms, recently an increasing number of continuous-time recurrent neural networks (RNNs) have been used to solve optimization problems. Essentially, these RNNs are continuous-time dynamical systems converging to an equilibrium that is also the optimizer of the problem. Consequently, much research effort has been devoted to  characterizing the stability of those systems and their convergence rates~\cite{KJA-LH-HU:58, JJH-DWT:85, AB-TRP:93, VC-AG-AD-GR-FB:23c, AD-VC-AG-GR-FB:23f}.
In fact, the use of continuous-time RNNs to solve optimization problems and research on stability conditions for these RNNs has gained considerable interest in a wide range of fields, see, e.g.,~\cite{HZ-ZW-DL:14, AB-JR-CJR:12, AB-CJR-JR:13b}.
An RNN designed to tackle the SR problem is the Locally Competitive Algorithm (LCA) introduced by Rozell et al.~\cite{CJR-DHJ-RGB-BAO:08}.
This network is a continuous-time Hopfield-like neural network~\cite{JJH:84} (HNN) of the form:
\begin{align}
\hopdot(t) = -\hop(t) + W \Psi(\hop(t))+\uh(t),
\label{eq:hopfield_nn}
\end{align}
with output $y(t) = \Psi(\hop(t))$. In~\eqref{eq:hopfield_nn} the state $\hop\in\R^n$ is usually interpreted as a membrane potential of the $n$ neurons in the network, $W\in\R^{n\times n}$ is a symmetric synaptic matrix, $\map{\Psi}{\R^n}{\R^n}$ is a diagonal activation function, and $\uh$ is an external stimulus.
Following~\cite{CJR-DHJ-RGB-BAO:08}, several results were established to analyze the properties of the LCA. Specifically, in~\cite{AB-CJR-JR:11} it is proven that, provided that the fixed point of the LCA is unique then, for a certain class of activation functions, the LCA globally asymptotically converges.
Then, in~\cite{AB-JR-CJR:12} it is shown that the fixed points of the LCA coincide with the solutions of the SR problem. Using a Lyapunov approach, under certain conditions on the activation function and on the solutions of the systems, it is also shown that the LCA converges to a single fixed point with exponential rate of convergence.
Various sparsity-based probabilistic inference problems are shown to be implemented via the LCA in~\cite{ASC-PG-CJR:12}.
In~\cite{AB-CJR-JR:13b} a technique using the Łojasiewicz inequality is used to prove convergence of both the output and state variables of the LCA.
\cite{AB-CJR-JR:13},~\cite{AB-CJR-JR:15},~\cite{MZ-CJR:13} focus on analyzing the LCA for the SR problem with $\ell_1$ sparsity-inducing penalty function. Specifically, the convergence rate is analyzed in~\cite{AB-CJR-JR:13}.
In~\cite{AB-CJR-JR:15} it is rigorously shown how the LCA can recover a time-varying signal from streaming compressed measurements.
Additionally, physiology experiments in~\cite{MZ-CJR:13} demonstrate that numerous response properties of non-classical receptive field (nCRF) can be reproduced using a model having the LCA as neural dynamics with an additional non-negativity  constraint enforced on the output to represent the instantaneous spike rate of neurons within the population.
We also note that, while the LCA is biologically inspired, as noted in, e.g.,~\cite{DL-CP-DBC:22} a biologically plausible network should exhibit non-negative states and this property is not guaranteed by the LCA.

Motivated by this, we consider positive SR problems, i.e., a class of SR problems with non-negativity constraints on the states and to tackle these problems we introduce the \emph{positive firing-rate competitive network} (\pfcn).
This is an FNN of the form (see, e.g.,~\cite{PD-LFA:05})
\begin{align}
\frdot(t) = - \fr(t) + \Psi(W\fr(t) + \ufr(t)),
\label{eq:firing_rate_nn}
\end{align}
with output $y(t) = \fr(t)$ and where the state $\fr \in\R^n$ is interpreted as the firing-rate of the $n$ neurons in the network, and, as in~\eqref{eq:hopfield_nn}, $W$ is the synaptic matrix, $\Psi$ is the activation function, and $\ufr$ is an external stimulus (or input).

The HNN~\eqref{eq:hopfield_nn} and FNN~\eqref{eq:firing_rate_nn} are known to be mathematical equivalent~\cite{KDM-FF:12} through suitably defined state and input transformations.
However, the input transformation is state dependent precisely when the synaptic matrix is rank deficient (as in sparse reconstruction) and, counter-intuitively, the transformation of solutions from HNN to FNN requires that the initial condition of the input depends on the initial condition of the state.
Moreover, the FNN might hold an advantage over the HNN in terms of biological plausibility in the following sense. When the activation function is non-negative, the positive orthant is forward-invariant, i.e., the state remains non-negative from non-negative initial conditions and is thus interpreted as a vector of firing-rates.
Therefore, even if the HNN state can be interpreted as a vector of membrane potentials, it is more natural to interpret negative (resp.~positive) synaptic connections as inhibitory (resp.~excitatory) in the FNN rather than the HNN.

To the best of our knowledge, the \pfcn is the first RNN to tackle positive SR problems and with our main results we characterize the behavior of this network, showing that it indeed solves this class of problems. Our analysis leverages contraction theory~\cite{WL-JJES:98} and, in turn, this allows us to also characterize the behavior of the \emph{firing-rate competitive network} (\fcn), i.e., a firing-rate version of the LCA, able to tackle the SR problem. Our use of contraction theory is motivated by the fact that contracting dynamics are robustly stable and enjoy many properties, such as certain types of input-to-state stability. For further details, we refer to the recent monograph~\cite{FB:23-CTDS} and to, e.g., works on recent applications of contraction theory in computational biology, neuroscience, and machine learning~\cite{GR-MDB-EDS:10a, AD-AVP-FB:21k, VC-FB-GR:22g, LK-ME-JJES:22}.
Our key technical contributions can then be summarized as follows:

\begin{enumerate}
\item\label{contrib:1} We propose, and analyze, the \emph{firing-rate competitive network} and the \emph{positive firing-rate competitive network} to tackle the SR and positive SR problem, respectively. 
First, we introduce a result relating the equilibria of the proposed networks to the optimal solutions of sparse reconstruction problems.
Then, we characterize the convergence of the dynamics towards the equilibrium. For the \pfcn, we also show that this is a positive system, i.e., if the system starts with non-negative initial conditions, its state variables remain non-negative. 
After characterizing the local stability and contractivity for the dynamics of our interest, with our main convergence result we prove that, under a standard assumption on the dictionary, our dynamics converges linear-exponentially to the equilibrium, in the sense that (in a suitably defined norm) the trajectory's distance from the equilibrium is initially upper bounded by a linear function and then convergence becomes exponential.  
We also give explicit expressions for the average linear decay rate and the time at which exponential convergence begins.

\item To achieve~\ref{contrib:1}, we propose a top/down normative framework for a biologically-plausible explanation of neural circuits solving sparse reconstruction and other optimization problems. To do so, we leverage tools from monotone operator theory~\cite{PLC-JCP:11,NP-SB:14} and, in particular, the recently studied \emph{proximal gradient dynamics}~\cite{SHM-MRJ:21, AD-VC-AG-GR-FB:23f}. 
This general theory explains how to transcribe a composite optimization problem into a continuous-time firing-rate neural network, which is therefore interpretable.

\item Our analysis of the \fcn and \pfcn dynamics naturally leads to the study of the convergence of \emph{globally-weakly and locally-strongly contracting systems}.
These are dynamics that are weakly infinitesimally contracting on $\R^n$ and strongly infinitesimally contracting on a subset of $\R^n$ (see Section~\ref{sec:contraction_theory} for the definitions).
We then conduct a comprehensive convergence analysis of this class of dynamics, which generalizes the linear-exponential convergence result for the \fcn and \pfcn to a broader setting. We also provide a useful technical result on the $\ell_2$ logarithm norm of upper triangular block matrices.
\item Finally, we illustrate the effectiveness of our results via numerical experiments. The code to replicate our numerical examples is available at \url{https://tinyurl.com/PFCN-for-Sparse-Reconstruction}.
\end{enumerate}
The rest of the paper is organized as follows. In Section 2, we provide some useful mathematical preliminaries: an overview of the SR problem, norms, and logarithmic norms definitions and results, and a review of contraction theory. In Section 3, we present the main results of the paper: we propose the \fcn and the \pfcn, establish the equivalence between the optimal solution of the SR problem and the equilibria of the \fcn, and prove the linear-exponential convergence behavior of our models.
In Section 4, we illustrate the effectiveness of our approach via a numerical example. In Section 5, we analyze the convergence of globally-weakly and locally-strongly contracting systems, showing linear-exponential convergence behavior of these systems.
We provide a final discussion and future prospects in Section 7. Finally, in the appendices, we provide instrumental results and review concepts useful for our analysis.

\section{Mathematical Preliminaries}
\subsection{Notation}
We denote by $\1_n$, $\0_n \in \R^n$ the all-ones and all-zeros vectors, respectively. We denote by $\ball{p}{\radius}:= \setdef{z\in\R^n}{\norm{z-x}_p\leq\radius}$ the ball of radius $\radius >0$ centered at some $x \in \R^n$ and whose distance with respect to (w.r.t.) the center is computed w.r.t. the norm $p$. We specify the center of $\ball{p}{\radius}$ when $x \neq \0_n$. We let $\diag{x} \in \R^{n \times n}$ be the diagonal matrix with diagonal entries equal to $x$ and $I_n$ be the $n \times n$ identity matrix.
For $A \in \R^{n \times n}$ we denote by $\rank(A)$ its rank, and by $\alpha(A) := \max \setdef{\operatorname{Re}(\lambda)}{\lambda \text{ eigenvalue of } A}$ its \emph{spectral abscissa}, where $\operatorname{Re}(\lambda)$ denotes the real part of $\lambda$. Given symmetric $A, B \in \R^{n \times n}$, we write $A \preceq B$ (resp. $A \prec B$) if $B-A$ is positive semidefinite (resp. definite). The function $\map{\ceil{ \ }}{\R}{\Z}$ is the \emph{ceiling function} and is defined by $\ceil{x} = \min\setdef{y \in \Z}{x \leq y}$. The \emph{subdifferential of $\map{g}{\R^n}{\R}$ at $x \in \R^n$} is the set $\partial g(x) := \setdef{z \in \R^n}{g(x) - g(y) \geq z^\top(x-y), \forall y \in \R^n}$.
Finally, whenever it is clear from the context, we omit to specify the dependence of functions on time $t$.

\subsection{The Sparse Reconstruction Problems}
\label{sec:sparse_representation}
Given a $m$-dimensional input $u\in \R^m$ (e.g., a $m$-pixel image), the \emph{sparse reconstruction} problem consists in reconstructing $u$ with a linear combination of a sparse vector $y\in \R^n$ and a \emph{dictionary} $\Phi\in \R^{m\times n}$ composed of $n$ (unit-norm) vectors $\Phi_i\in \R^m$ (see Figure~\ref{fig:frlca}.b)).
Following~\cite{BAO-DJF:97}, \emph{sparse reconstruction problems} can be formulated as follows:
\begin{align}
\label{eq:sparse_approx_C}
\min_{y \in \R^n} \Bigl( E(y) := \frac{1}{2}\big\|u- \Phi y\big\|^2_2 + \lambda S(y)\Bigr),
\end{align}
where $\lambda \geq 0$ is a scalar parameter that controls the trade-off between accurate reconstruction error (the first term) and sparsity (the second term). Indeed, in~\eqref{eq:sparse_approx_C} $\map{S}{\R^n}{\R}$ is a non-linear cost function that induces sparsity and is typically assumed to be convex and separable across indices, i.e., $S(y) = \sum_{i = 1}^n s(y_i)$, for all $y \in \R^n$, with $\map{s}{\R}{\R}$.
Using the definition of $\ell_2$ norm we can write
$$
E(y) = \frac{1}{2}(\trasp{u}u- 2\trasp{u}\Phi y + \trasp{y}\trasp{\Phi}\Phi y) + \lambda S(y).
$$
The matrix $\trasp{\Phi}\Phi \in \R^{n\times n}$ is known as $\emph{Gramian matrix}$ of $\Phi$.
Note that, when $S$ is convex and $\rank(\Phi) = n$, the objective function $E(y)$ is strongly convex, therefore~\eqref{eq:sparse_approx_C} admits a unique solution. While, when $\rank(\Phi) < n$, $E(y)$ is not strongly convex, leading to multiple solutions. Specifically, when $n > m$, we must have $\rank(\Phi) < n$. SR problems focus on the underdetermined case, i.e., when $n \gg m$.

A common choice of $S$ is the $\ell_1$ norm, resulting in the following formulation of~\eqref{eq:sparse_approx_C}, known as \emph{basis pursuit denoising} or \emph{lasso}:
\begin{align}
\label{eq:sparse_approx_l1}
\min_{y \in \R^n} \Bigl(\Elasso(y) := \frac{1}{2}\big\|u- \Phi y\big\|^2_2 + \lambda \|y\|_1\Bigr).
\end{align}
For problem~\eqref{eq:sparse_approx_l1}, accurate reconstruction of $u$ is possible under the condition that $u$ is sparse enough and the dictionary satisfies the following:
\bd[$k$-sparse vector and RIP~\cite{EJC-TT:07}]
Let $k < n$ be natural numbers. A vector $x \in \R^n$ is \emph{$k$-sparse} if it has at most $k$ non-zero entries. A matrix $\Phi \in \R^{n \times m}$ satisfies the \emph{restricted isometry property (RIP)} of order $k$ if there exist a constant $\delta \in [0, 1)$, such that for all $k$-sparse $x \in \R^{n}$ we have
\beq
\label{eq:rip}
(1-\delta)\normtwo{x}^2 \leq \normtwo{\Phi x}^2 \leq (1+\delta) \normtwo{x}^2.
\eeq
The \emph{order-$k$ restricted isometry constant} $\delta_k$ is the smallest $\delta$ such
that~\eqref{eq:rip} holds.
\ed

We are particularly interested in~\eqref{eq:sparse_approx_l1} when this has non-negative constraints. We term this problem the \emph{positive sparse reconstruction problem} and the goal is to reconstruct an input $u$ using a linear combination of a non-negative and sparse vector $y\in \R^n_{\geq0}$ and a unit-norm \emph{dictionary} $\Phi\in \R^{m\times n}$. Formally, the positive sparse reconstruction problem can be stated as follows:
\beq
\label{eq:positive_E_lasso+constraint}
\begin{split}
\min_{y \in \R^n} &\, \frac{1}{2}\big\|u- \Phi y\big\|^2_2 + \lambda \|y\|_1,\\
\text{s.t.}&\, y \in \R^n_{\geq0}.
\end{split}
\eeq
The minimization problem~\eqref{eq:positive_E_lasso+constraint} can equivalently be written as the unconstrained optimization problem
\beq
\label{eq:positive_E_lasso_unconstrained}
\min_{y \in \R^n} \frac{1}{2}\big\|u- \Phi y\big\|^2_2 + \lambda \|y\|_1 + \iota_{\R^n_{\geq0}}(y),
\eeq
where $\map{\iota_{\R^n_{\geq0}}}{\R^n}{[0,+\infty]}$ is the \emph{zero-infinity indicator function on $\R^n_{\geq0}$} and is defined by $\iota_{\R^n_{\geq0}}(x) = 0$ if $x \in \R^n_{\geq0}$ and $\iota_{\R^n_{\geq0}}(x) = +\infty$ otherwise.

We note that problem~\eqref{eq:positive_E_lasso_unconstrained} can be formally written as problem~\eqref{eq:sparse_approx_C} when the sparsity inducing cost in~\eqref{eq:sparse_approx_C} is $S_1(y):=\|y\|_1 + \frac{1}{\lambda}\iota_{\R^n_{\geq0}}(y) = \sum_{i=1}^n \bigl(y_i + \frac{1}{\lambda}\iota_{\R_{\geq0}}(y_i)\bigr)$, where we used the fact that $y$ must belong to $\R_{\geq 0}^n$. Also, for our derivations, it is useful to introduce the scalar function $s_1(y_i) := y_i + \frac{1}{\lambda}\iota_{\R_{\geq0}}(y_i)$.
\subsection{Norms and Logarithmic Norms}
\label{Norms and logarithmic norms}
Given two vector norms $\norm{\cdot}_{\alpha}$ and $\norm{\cdot}_{\beta}$ on $\R^n$ there exist positive
\emph{equivalence coefficients} $k_\alpha^\beta>0$ and $k_\beta^\alpha>0$ such that
\beq
\label{eq:equivalence_coeff_norms}
\norm{x}_{\alpha}\leq\kab\norm{x}_{\beta}, \quad \norm{x}_{\beta}\leq\kba\norm{x}_{\alpha}, \quad \textup{for all } x\in\R^n.
\eeq

For later use, we give the following

\smallskip
\bd[Equivalence ratio between two norms]
\label{def:equivalence_ratio}
Given two norms $\norm{\cdot}_{\alpha}$ and $\norm{\cdot}_{\beta}$, let $k_{\alpha}^{\beta}$ and $k_{\beta}^{\alpha}$ be the minimal coefficients satisfying~\eqref{eq:equivalence_coeff_norms}. The \emph{equivalence ratio between $\norm{\cdot}_{\alpha}$ and $\norm{\cdot}_{\beta}$} is $\prodeqnorm_{\alpha,\beta} :=\kab\kba$.
\ed

\smallskip
Let $\| \cdot \|$ denote both a norm on $\R^n$ and its corresponding induced matrix norm on $\R^{n \times n}$. %Specifically, we are interested in $\ell_2$ norms. 
Given $x\in \R^n$ and $A \in \R^{n \times n}$ we recall that the $\ell_2$ vector norm and matrix norm are, respectively, ${\displaystyle \norm{x}_2^2 = x^\top x}$, $\displaystyle {\norm{A}_2^2 = \subscr{\lambda}{max}(A^\top A)}$. The \emph{logarithmic norm} (log-norm) induced by the $\ell_2$ norm is $\displaystyle \lognorm{A}{2} = \subscr{\lambda}{max}\bigl((A + A^\top)/{2}\bigr)$.
For an invertible $\Q \in \R^{n\times n}$, the $\Q$-weighted $\ell_2$ matrix norm is $\norm{A}_{2,Q} = \norm{Q A Q^{-1}}_2$. The corresponding log-norm is $ \wlognorm{2}{Q}{A} = \lognorm{Q A Q^{-1}}{2}$. %For the sake of simplicity of notation, from now and throughout the rest of the paper we use the notation  $\norm{\cdot}_{2,Q} :=  \norm{\cdot}_{Q}$ and $ \wlognorm{2}{Q}{\cdot} = \lognorm{\cdot}{Q}$ for the weighted Euclidean norms and log-norms, respectively.
\subsection{Contraction Theory for Dynamical Systems}
\label{sec:contraction_theory}
Consider a dynamical system 
\beq
\label{eq:dynamical_system}
\dot{x}(t) = f\bigl(t,x(t)\bigr),
\eeq
where $\map{f}{\R_{\geq 0} \times C}{\R^n}$, is a smooth nonlinear function with $C\subseteq \R^n$ forward invariant set for the dynamics. We let $t \mapsto \odeflowtx{t}{x(0)}$ be the flow map of~\eqref{eq:dynamical_system} starting from initial condition $x(0)$.
First, we give the following:
\bd[Contracting systems~\cite{FB:23-CTDS}] \label{def:contracting_system}
Given a norm $\norm{\cdot}$ with associated log-norm $\mu$, a smooth function $\map{f}{\R_{\geq 0} \times C}{\R^n}$, with $C \subseteq \R^n$ $f$-invariant, open and convex, and a constant $c >0$ ($c = 0)$ referred as \emph{contraction rate}, $f$ is $c$-strongly (weakly) infinitesimally contracting on $C$ if
\beq\label{cond:contraction_log_norm}
\mu\bigl(Df(t, x)\bigr) \leq -c,  \textup{ for all } x \in C  \textup{ and } t\in \R_{\geq0},
\eeq
where $Df(t,x) := \partial f(t,x)/\partial x$ is the Jacobian of $f$ with respect to $x$.
\ed
One of the benefits of contraction theory is that it enables the study of the convergence behavior of the flow map with a single condition.
Specifically, if $f$ is contracting, for any two trajectories $x(\cdot)$ and $y(\cdot)$ of~\eqref{eq:dynamical_system} it holds
$$\|\odeflowtx{t}{x(0)} - \phi_t(y(0))\| \leq \e^{-ct}\|x(0) -y(0)\|, \quad \textup{ for all } t \geq 0,$$
i.e., the distance between the two trajectories converges exponentially with rate $c$ if $f$ is $c$-strongly infinitesimally contracting, and never increases if $f$ is weakly infinitesimally contracting.

Strongly infinitesimally contracting systems enjoy many useful properties. Notably, initial conditions are exponentially forgotten~\cite{WL-JJES:98}, time-invariant dynamics admit a unique globally exponential stable equilibrium~\cite{WL-JJES:98} (see Figure~\ref{fig:contracting_systems}.a)), and enjoy highly robust behaviors~\cite{GR-MDB-EDS:10a, SX-GR-RHM:21}.
These properties do not generally extend to weakly infinitesimally contracting systems. Nevertheless, these systems still enjoy numerous useful properties, such as the so-called dichotomy property~\cite{SJ-PCV-FB:19q}. This property states that if a weakly infinitesimally contracting system on $C$ has no equilibrium point in $C$, then every trajectory starting in $C$ is unbounded (see Figure~\ref{fig:contracting_systems}.b)); otherwise, if the system has at least one equilibrium, then every trajectory starting in $C$ is bounded (see Figure~\ref{fig:contracting_systems}.c)).

\begin{figure}[!ht]
\centering
\includegraphics[scale = 0.36]{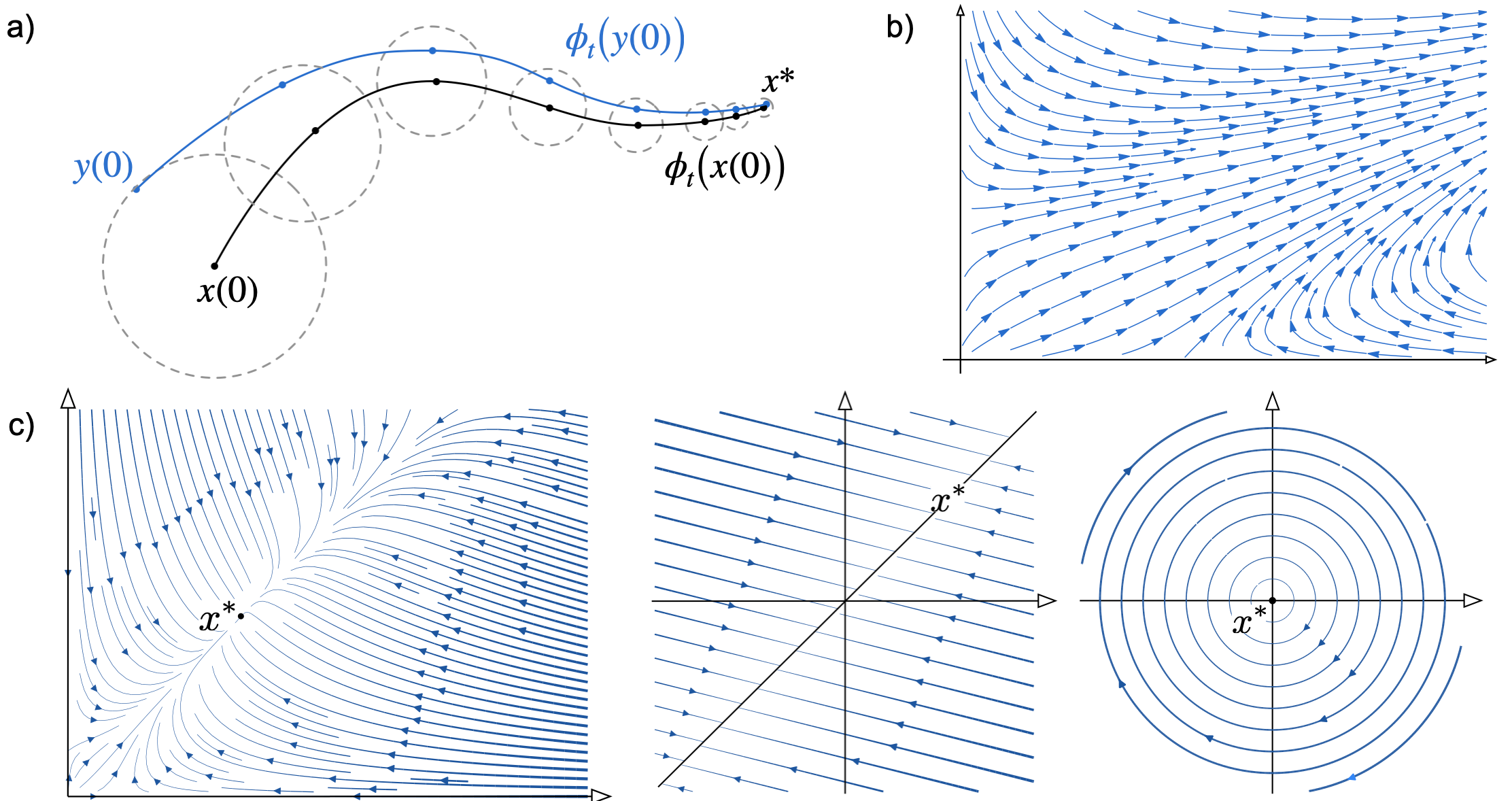}
\caption{Strongly infinitesimally contracting systems: a) the distance between any two trajectories converges exponentially to the unique equilibrium point $\xstar$. Illustration of the dichotomy property of weakly contracting systems: b) the system has no equilibrium and every trajectory is unbounded or c) there exists at least one equilibrium and every trajectory is bounded. Images reused with permission from~\cite{FB:23-CTDS}.}
\label{fig:contracting_systems}
\end{figure}
Of particular interest is the case of nonsmooth map $f$. In~\cite[Theorem 6]{AD-AVP-FB:22q} condition~\eqref{cond:contraction_log_norm} is generalized for locally Lipschitz function, for which the Jacobian $Df$ exists almost everywhere (a.e.) in $C$. Specifically, if $f$ is locally Lipschitz, then  $f$ is infinitesimally contracting on $C$ if condition~\eqref{cond:contraction_log_norm} holds for a.e. $x \in C$ and $t\in \R_{\geq0}$.

Finally, we recall the following result in~\cite[Corollary 1.(i)]{VC-AG-AD-GR-FB:23c} on the weakly infinitesimally contractivity of \FNN~\eqref{eq:firing_rate_nn} that will be useful for our analysis.
\begin{lem}[Weakly contractivity of the \FNN]\label{thm:contractivity_fnn}
Consider the \FNN~\eqref{eq:firing_rate_nn} with symmetric weight matrix $W$, and with activation function $\Psi$ being Lipschitz and slope restricted in $[0,1]$. If $\lmax = 1$, then the \FNN is weakly infinitesimally contracting with respect to some weighted Euclidean norm, say $\norm{\cdot}_{2,\qw}$.\footnote{The explicit expression for the matrix $\qw$ is given in Appendix~\ref{apx:weight_matrix_D}.}
\end{lem}

\section{Main Results}
In this section, we present the main results of this paper. Specifically, we first introduce a family of continuous-time FNNs, i.e., the \fcn, to tackle the SR problem in~\eqref{eq:sparse_approx_C} and give a result relating the equilibria of the former to the optimal solutions of the latter. Then, we consider the SR problem with non-negativity constraints in~\eqref{eq:positive_E_lasso_unconstrained} and propose an FNN network, i.e., the \pfcn, to also tackle such problem.

\subsection{Firing-rate Neural Networks For Solving Sparse Reconstruction Problems}
\label{subsec:fnn_for_sr_problem}
The SR problems introduced in Section~\ref{sec:sparse_representation} naturally arise in the context of visual information processing. For example, as illustrated in Figure~\ref{fig:frlca}.a) for mammalians, the visual sensory input data $u \in \R^m$ is encoded by the receptive fields of simple cells in V1 using only a small fraction of active (sparse) neurons. Formally (see Figure \ref{fig:frlca}.b)), the input signal $u$ is reconstructed through a linear combination of an overcomplete matrix $\Phi \in \mathbb{R}^{m \times n}$ and a sparse vector $y \in \R^n$. The \fcn and \pfcn introduced in this paper to tackle the SR and positive SR problems are schematically illustrated in  Figure~\ref{fig:frlca}.c). Namely, each hidden node, or neuron in what follows, $x_i$ receives as stimulus the similarity score between the input signal $u\in \R^m$ and the dictionary element $\Phi_i\in \R^n$ and, collectively, all the hidden neurons give as output a sparse (non-negative) vector $y = x \in \R^n$.

\begin{figure}[!ht]
\centering
\includegraphics[width=0.8\linewidth]{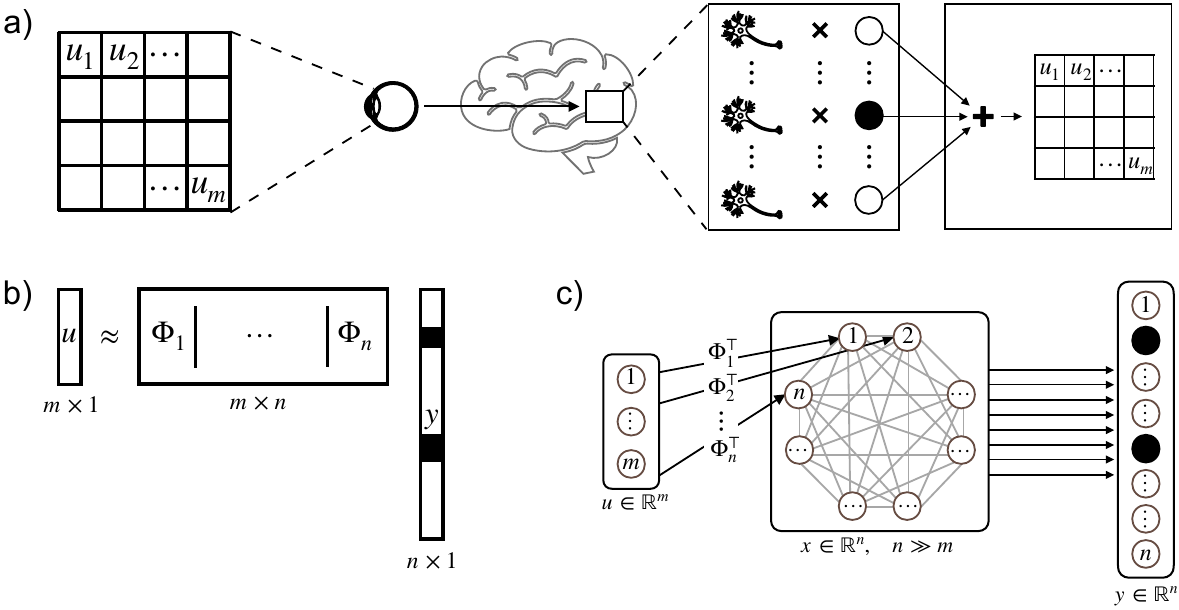}
\caption{The visual sensory input data $u \in \R^m$ is encoded by the receptive fields of simple cells in the mammalian visual cortex (V1) using only a small fraction of active (sparse) neurons. Formally, b) the input $u$ is reconstructed by a linear combination of an overcomplete ($n \gg m$) set of features $\Phi_i \in \R^n$ and sparse neurons $y \in \R^n$. c) Block scheme of the proposed (positive) firing-rate competitive network. The hidden node $x_i$ receives as stimulus the similarity score between the input signal $u\in \R^m$ and the dictionary element $\Phi_i\in \R^n$ and collectively all hidden neurons give as output a sparse (non-negative) vector $y = x \in \R^n$. Images reused with permission from~\cite{FB:23-CTDS}.}
\label{fig:frlca}
\end{figure}

\smallskip
To devise our results, we make use of the following standard assumption on the sparsity-inducing cost function $S$.
\smallskip
\begin{assumption}
\label{ass:actvation_function}
The function $S$ is convex, closed, proper\footnote{We refer to Appendix~\ref{apx:proximal_operator} for the definition of those notions.}, and separable across the indices, i.e., $S(y) = \sum_{i = 1}^n s(y_i)$, for all $y \in \R^n$, where $\map{s}{\R}{\R}$ is a convex, closed, and proper scalar function.
\end{assumption}
\smallskip

In order to transcribe the SR problem in~\eqref{eq:sparse_approx_C} into an interpretable continuous-time firing-rate neural network, we leverage the theory of \emph{proximal operators}.
The proximal operator of a convex function is a natural extension of the notion of projection operator onto a convex set. This concept has gained increasing significance in various fields, particularly in signal processing and optimization problems~\cite{PLC-JCP:11, AB:17}.
We refer to Appendix~\ref{apx:proximal_operator} for a self-contained primer on proximal operators, including the definition of the continuous-time proximal gradient dynamics~\eqref{apx:eq:prox:gradient}.

The SR problem~\eqref{eq:sparse_approx_C} is a special instance of the composite minimization problem~\eqref{apx:eq:composite_problem} in Appendix~\ref{apx:proximal_operator} with $f(x) := \frac{1}{2}\big\|u- \Phi y\big\|^2_2$ and $g(x) := \lambda S(y)$.
Therefore, to tackle problem~\eqref{eq:sparse_approx_C} we introduce the following special instance of proximal gradient dynamics, the \emph{firing-rate competitive network} (\fcn): 
\beq
\label{eq:frlca-x_dot_general}
\dot{x}(t) = -x(t) +\prox{\lambda S}{\bigl((I_n-\trasp{\Phi}\Phi)x(t) + \trasp{\Phi} u(t)\bigr)},
\eeq
with output $y(t) = x(t)$. 
This dynamics is schematically illustrated in Figure~\ref{fig:frlca}.c).
In~\eqref{eq:frlca-x_dot_general}, the term $\trasp{\Phi} u(t)$ is the input to the \fcn and it captures the similarity between the input signal and the dictionary elements, while the term $\bigl(I_n-\trasp{\Phi}\Phi\bigr)x(t)$ models the recurrent interactions between the neurons. These interactions implement competition between nodes to represent the stimulus.
Additionally, we note that in~\eqref{eq:frlca-x_dot_general} the particular form of the activation function is linked to the sparsity-inducing term in~\eqref{eq:sparse_approx_C}, $\lambda S$, via the proximal operator.
To be precise, the activation function is the proximal operator of $\lambda S$ computed at the point $x - \nabla\bigl(\frac{1}{2}\big\|u- \Phi x\big\|^2_2\bigr) = x - \Phi^\top \Phi x + \Phi^\top u$.

We now make explicit how the dynamics~\eqref{eq:frlca-x_dot_general} reads for the $\ell_1$ SR problem in~\eqref{eq:sparse_approx_l1} and the positive SR problem in~\eqref{eq:positive_E_lasso_unconstrained}.

For the lasso problem~\eqref{eq:sparse_approx_l1}, the sparsity-inducing cost function is $S(x) = \norm{x}_1$. This function is convex, separable and $s(x_i) = \abs{x_i}$, for all $x_i \in \R$. Moreover, it is well know (see, e.g.,~\cite{NP-SB:14}) that for any $x \in \R^n$, the proximal operator of $\lambda \norm{x}_1$ is $\prox{\lambda \norm{x}_1} = \softt{\lambda}{x}$, where $\map{\operatorname{soft}_{\lambda}}{\R^n}{\R^n}$ %, illustrated in Figure~\ref{fig:plot_soft},
is the \emph{soft thresholding function} defined by $(\softt{\lambda}{x})_i= \softt{\lambda}{x_i}$, and the map $\map{\operatorname{soft}_{\lambda}}{\R}{\R}$ is defined by
\[
\softt{\lambda}{x_i} =
\begin{cases}
0 & \textup{ if } \abs{x_i} \leq \lambda, \\
x_i - \lambda\sign{x_i} & \textup{ if } \abs{x_i} > \lambda,
\end{cases}
\]
with $\map{\operatorname{sign}}{\R}{\{-1,0,1\}}$ being the \emph{sign function} defined by $\sign{x_i} := -1$ if $x_i<0$, $\sign{x_i} := 0$ if $x_i=0$, and $\sign{x_i} :=1$ if $x_i>0$.
Now and throughout the rest of the paper, we adopt a slight abuse of notation by using the same symbol to represent both the scalar and vector form of the activation function.
The corresponding \fcn~\eqref{eq:frlca-x_dot_general} for the lasso problem~\eqref{eq:sparse_approx_l1} is therefore:
\begin{align}
\label{eq:fr-sa_C-x_dot}
\dot x(t) &= - x(t) + \soft{\lambda}\bigr((I_n-\trasp{\Phi}\Phi)x(t) + \trasp{\Phi} u(t)\bigl).
\end{align}
Remarkably, the dynamics~\eqref{eq:fr-sa_C-x_dot} is the firing-rate version of the LCA designed for tackling the lasso problem~\eqref{eq:sparse_approx_l1}, which is a continuous-time Hopfield-like neural network of the form~\cite{CJR-DHJ-RGB-BAO:08}:
\begin{align}
\label{eq:lca-soft-xi}
\dot x(t) &= - x(t) + \bigr(I_n-\trasp{\Phi}\Phi\bigl) \soft{\lambda}(x(t)) + \trasp{\Phi} u(t),
\end{align}
with output $y(t) = \soft{\lambda}(x(t))$.

Next, we define the \fcn that solves the positive SR problem~\eqref{eq:positive_E_lasso_unconstrained}. For this purpose, we need to determine the proximal operator of ${\lambda S_1(x) = \sum_{i=1}^n \bigl( \lambda x_i + \iota_{\R_{\geq0}}(x_i)\bigr)}$. We have -- see Lemma~\ref{apx:lem:prox_of_relu} in Appendix~\ref{apx:proximal_operator} for the mathematical details -- that $\prox{\lambda S_1}(x) = \relu(x - \lambda \1_n)$, for all $x \in \R^n$, where $\map{\relu}{\R}{\R}$ is the \emph{(shifted) ReLU function} defined by
\[
\relu{(x_i - \lambda)} =
\begin{cases}
0 & \textup{ if } x_i \leq \lambda, \\
x_i - \lambda & \textup{ if } x_i > \lambda.
\end{cases}
\]

Thus, the \fcn~\eqref{eq:frlca-x_dot_general} that solves the positive SR problem~\eqref{eq:positive_E_lasso_unconstrained} is given by
\begin{equation}
\label{eq:pfr-lca_x_dot}
\begin{split}
\dot{x}(t) &= -x(t) +\relu{\bigl((I_n-\trasp{\Phi}\Phi)x(t) + \trasp{\Phi} u(t) - \lambda \1_n\bigr)}:= \fpfcn(x),
\end{split}
\end{equation}
with output $y(t) = x(t)$.
\begin{comment}
In components~\eqref{eq:pfr-lca_x_dot} reads
\begin{align}
\label{eq:pfr-lca_x_idot}
\dot x_i &= - x_i + \relu \Bigl(-\sum_{j \neq i, j = 1}^n\trasp{\Phi_i} \Phi_j x_j + \trasp{\Phi_i}u - \lambda\Bigr).
\end{align}
\end{comment}
We call these dynamics \emph{positive firing-rate competitive network} (\pfcn). A key property of the \pfcn is that it is a \emph{positive system}; i.e., given a non-negative initial state, the state variables are always non-negative (we refer to Appendix~\ref{app:positive_system} for a rigorous proof of this statement). This is a desirable property that can be useful to effectively model both excitatory and inhibitory synaptic connections. In fact, in the \pfcn the nature of excitatory and inhibitory recurrent interactions, described by the term $-\sum_{j \neq i, j = 1}^n\trasp{\Phi_i} \Phi_j x_j$, only depends on the sign of the weights. Specifically, the recurrent interaction between two nodes, say them $i$ and $j$, is inhibitory if $-\sum_{j \neq i, j = 1}^n\trasp{\Phi_i} \Phi_j< 0$, and excitatory if $-\sum_{j \neq i, j = 1}^n\trasp{\Phi_i} \Phi_j > 0$.

\smallskip
Finally, for later use, we note that the Jacobian of $\fpfcn$ exists a.e. by Rademacher’s theorem, and now and throughout the rest of the paper we denote by $\nondiffset{f} \subset \R^n$ the measure zero set of points where the function $\fpfcn$ is not differentiable.

\subsection{Analysis of the Proposed Networks}
We now investigate the key properties of the models introduced above. Intuitively, the results presented in this section can be informally summarized via the following:
\begin{metatheorem}[Summary of the main results]
\label{metatheorem}
The trajectories of the \pfcn~\eqref{eq:pfr-lca_x_dot} are bounded. Additionally, if the dictionary $\Phi$ is RIP, then:
\begin{enumerate}
    \item the \pfcn converges to an equilibrium point that is also the optimal solution of the positive SR problem~\eqref{eq:positive_E_lasso_unconstrained};
    \item the convergence is linear-exponential, in the sense that the trajectory's distance from the equilibrium point initially decays at worst linearly, and then, after a transient, exponentially.
\end{enumerate}
\end{metatheorem}

\begin{figure}[h!]
\centering
\includegraphics[width=0.9\linewidth]{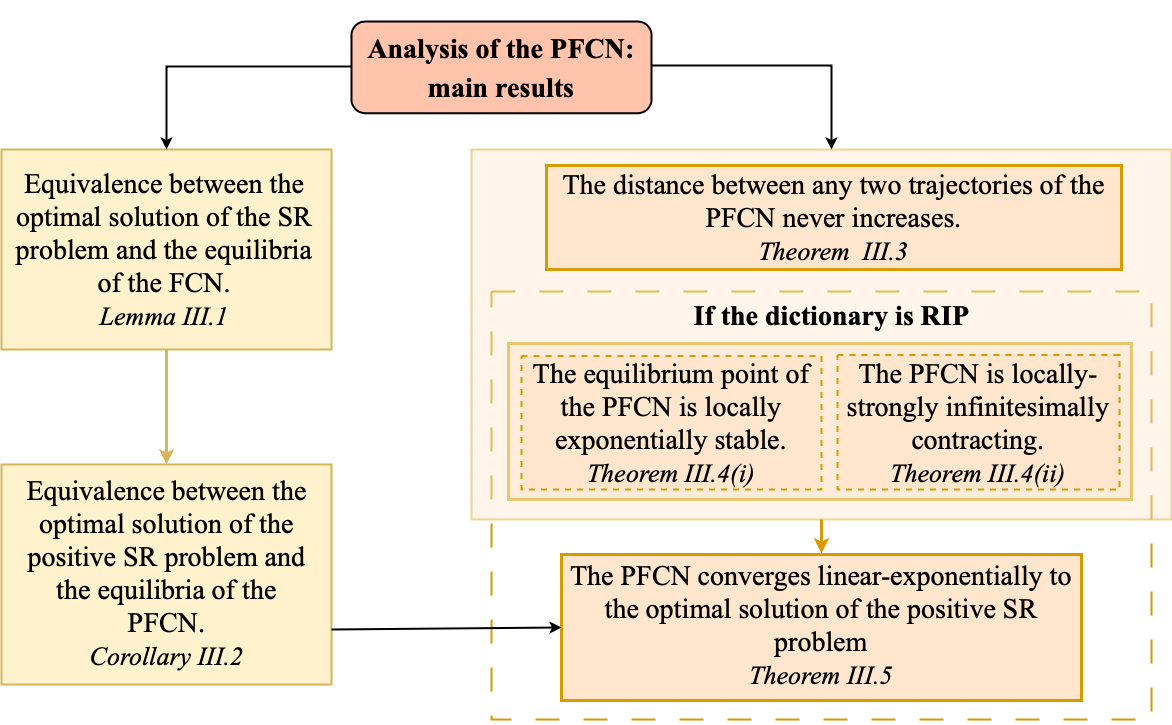}
\caption{Schematic diagram summarizing the main results and their assumptions. With Theorem~\ref{thm:GLin-ExpS_of_pfrlcn} we show that the \pfcn~\eqref{eq:pfr-lca_x_dot} exhibits linear-exponential convergence towards the optimal solution of the positive SR problem~\eqref{eq:positive_E_lasso_unconstrained}. The result follows from: (i) establishing a link between the optimal solution of~\eqref{eq:positive_E_lasso_unconstrained} and the equilibria of~\eqref{eq:pfr-lca_x_dot}; (ii) characterizing contractivity of~\eqref{eq:pfr-lca_x_dot}.}
\label{fig:metatheorem}
\end{figure}

The assumptions, the results, and their links towards building the claims in the informal statement are also summarized in Figure~\ref{fig:metatheorem}. Specifically, we first show that the equilibria of both the \fcn~\eqref{eq:frlca-x_dot_general} and \pfcn~\eqref{eq:pfr-lca_x_dot} are the optimal solutions of~\eqref{eq:sparse_approx_C} and~\eqref{eq:sparse_approx_l1}, respectively (Lemma~\ref{lem:opt_sol-eq_point} and Corollary~\ref{cor:pfcn_equilibria}). Then, we show that the distance between any two trajectories of the \pfcn never increases (Theorem~\ref{thm:weak-contractivity_pfrlcn}). Moreover, we show that if the dictionary $\Phi$ is RIP, then the equilibrium point for the \pfcn is not only locally exponentially stable but it is also strongly contracting in a neighborhood of the equilibrium (Theorem~\ref{thm:loc_exp_stab_loc_contractivity}). These results then lead to Theorem~\ref{thm:GLin-ExpS_of_pfrlcn}, where we show that the \pfcn~\eqref{eq:pfr-lca_x_dot} has a linear-exponential convergence behavior. That is, the distance between any trajectory of the \pfcn and its equilibrium is upper bounded, up to some \emph{linear-exponential crossing time}, say $\tld$, by a decreasing linear function. Then, for all $t>\tld$, the distance is upper bounded by a decreasing exponential function (see Figure~\ref{fig:lin_exp_schematic} for an illustration of this behavior).
To streamline the presentation, we provide explicit derivations for the \pfcn~\eqref{eq:pfr-lca_x_dot}. However, the analysis can be adapted for the \fcn~\eqref{eq:frlca-x_dot_general} (see Remark~\ref{rem:general_analysis} for the precise conditions).

\begin{SCfigure}[10][!ht]\centering
\includegraphics[width=.52\linewidth]{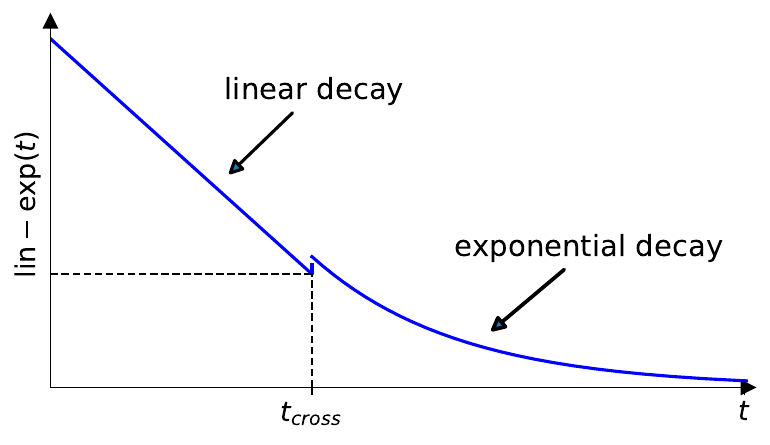}
\hspace{.03\linewidth}
\caption{Schematic representation of the linear-exponential convergence behavior exhibited by the \pfcn. The distance of the trajectory from the equilibrium point is upper bounded by a function that decreases linearly with time until $\tld$ and then exponentially for all $t>\tld$. While the solution of the \pfcn is continuous, a bounded jump in the upper bound we obtain might occur at time $\tld$.}
\label{fig:lin_exp_schematic}
\end{SCfigure}

\subsubsection{Relating the \fcn and \pfcn with SR Problems}
With Lemma~\ref{lem:opt_sol-eq_point} we show that a given vector is the optimal solution of~\eqref{eq:sparse_approx_C} if and only if this is also an equilibrium of the \fcn~\eqref{eq:frlca-x_dot_general}. Corollary~\ref{cor:pfcn_equilibria}, which follows from Lemma \ref{lem:opt_sol-eq_point}, relates the optimal solutions of~\eqref{eq:positive_E_lasso_unconstrained}  with the equilibria of the \pfcn~\eqref{eq:pfr-lca_x_dot}.

\begin{lem}[Linking the optimal solutions of the SR problem and the equilibria of the \fcn]
\label{lem:opt_sol-eq_point}
The vector $x^{*}\in\R^n$ is an optimal solution of the SR problem~\eqref{eq:sparse_approx_C} if and only if it is an equilibrium point of the \fcn~\eqref{eq:frlca-x_dot_general}.
\end{lem}
\begin{proof}
The necessary and sufficient condition for $\xstar \in\R^n$ to be a solution of %the minimization 
problem~\eqref{eq:sparse_approx_C} is
\[
\0_n \in \partial E(\xstar) = \Phi^\top \Phi \xstar - \Phi^\top u + \lambda \partial S(\xstar) := \OG(\xstar) + \lambda \partial S(\xstar),
\]
where we have introduced the function $\OG : \R^n \to \R^n$ defined as $\OG(x) = \Phi^\top \Phi x - \Phi^\top u$.
Note that $\OG$ is a linear function of $x$, therefore it is Lipschitz, and $\OG(x) = \nabla\bigl(\frac{1}{2}\big\|u- \Phi x\big\|^2_2\bigr)$. That is, $\OG$ is the gradient (w.r.t. $x$) of a convex function, and thus it is monotone.
Moreover, by Assumption~\ref{ass:actvation_function}, the function $S$ is convex, closed, and proper, and therefore, so it is $\lambda S$. Then, by applying the result in~\cite[Proposition 4]{AD-VC-AG-GR-FB:23f} (picking $\gamma= 1$, $\OF = \OG$, and $g = \lambda S$ in such proposition), we have that $\0_n \in \bigl(\OG + \partial \lambda S\bigr)(\xstar)$ if and only if $\xstar$ is an equilibrium of $\dot{x} = - x + \prox{\lambda S}\bigl(x - \OG(x)\bigr) = -x + \prox{\lambda S}\bigl((I_n - \Phi^\top \Phi) x + \Phi^\top u\bigr)$. This concludes the proof.
\end{proof}

We can then state the following:

\begin{cor}[Linking the optimal solutions of the positive SR problem and the equilibria of the \pfcn]\label{cor:pfcn_equilibria}
\label{cor:opt_sol-eq_point_pfcn}
The vector $x^{*}\in\R^n$ is an optimal solution of the positive SR problem~\eqref{eq:positive_E_lasso_unconstrained} if and only if it is an equilibrium point of the \pfcn~\eqref{eq:pfr-lca_x_dot}.
\end{cor}
\begin{proof} The proof, which follows the arguments used to prove Lemma~\ref{lem:opt_sol-eq_point}, is omitted for brevity.
\end{proof}
\smallskip

\subsubsection{Convergence Analysis}
We now present our convergence analysis for the \pfcn. In doing so, we introduce here the main convergence results and we refer to the Methods section and to the appendices for technical instrumental results. We start with the following:
\smallskip
\bd[Active and inactive neuron]\label{def:active_neurons}
Given a neural state $\xstar\in\R^n$, an input $u \in \R^m$, and a parameter $\lambda > 0$, the $i$-th neuron is \emph{active} if $\relu\bigl(((I_n - \trasp{\Phi}\Phi)\xstar + \trasp{\Phi}u - \lambda \1_n)_i\bigr) \neq 0$, \emph{inactive} if $\relu\bigl(((I_n - \trasp{\Phi}\Phi)\xstar + \trasp{\Phi}u - \lambda \1_n)_i\bigr) = 0$.
\ed

\smallskip
\begin{rem}
The definition of \emph{active} and \emph{inactive} neuron/node in our model aligns with the definitions provided in~\cite{AB-JR-CJR:12} for the LCA.
Specifically, as in~\cite{AB-JR-CJR:12}, for an equilibrium point $\xstar \in \mathbb{R}^n$, the activation function in our model is also composed of two operational regions. Namely: (i) one region characterized by having $(I_n - \trasp{\Phi}\Phi)\xstar + \trasp{\Phi}u - \lambda \1_n$ below zero, in which case the output $y$ is zero, as the system is at the equilibrium $\xstar = \relu\bigl((I_n - \trasp{\Phi}\Phi)\xstar + \trasp{\Phi}u - \lambda \1_n\bigr)$. (ii) one region characterized by having $(I_n - \trasp{\Phi}\Phi)\xstar + \trasp{\Phi}u - \lambda \1_n$ above zero, in which case $y$ is strictly increasing with the state $x$.
\end{rem}

\smallskip
\begin{remark}
\label{rem:general_analysis}
We present the convergence analysis for the \pfcn. However, our results can be extended to any \fcn~\eqref{eq:frlca-x_dot_general}, whose proximal operator is Lipschitz and slope restricted in $[0,1]$. For the \fcn, the $\relu$ in Definition \ref{def:active_neurons} is replaced by $\prox{\lambda S}{ }$. For example, our convergence analysis can be extended to the firing-rate version of the LCA tackling problem~\eqref{eq:sparse_approx_l1}, i.e. the dynamics~\eqref{eq:fr-sa_C-x_dot}, being $\soft{\lambda}$ Lipschitz and slope restricted in $[0,1]$.

%For example, for the firing-rate version of the LCA tackling problem~\eqref{eq:sparse_approx_l1}, i.e. the dynamics~\eqref{eq:fr-sa_C-x_dot},  this means that $\soft{\lambda}$ must be slope restricted in $[0,1]$.
\end{remark} 

\smallskip
We now show that the distance between any two trajectories of the \pfcn never increases (see Figure \ref{fig:metatheorem}). We do so by proving that the \pfcn is weakly infinitesimally contracting.

\bt[Global weak contractivity of the \pfcn]\label{thm:weak-contractivity_pfrlcn}
The \pfcn~\eqref{eq:pfr-lca_x_dot} is weakly infinitesimally contracting on $\R^n$ with respect to the weighted norm $\norm{\cdot}_{2,\qw}$.\footnote{the explicit expression of $\qw \in \R^{n \times n}$ is given in~\eqref{eq:weight_matrix_wc} of Appendix~\ref{apx:weight_matrix_D}}
\et
\begin{proof}
First, we note that the activation function $\relu$ is Lipschitz with constant $1$ and slope restricted in $[0,1]$. Indeed the partial derivative of $\relu$, ${\partial \bigl(\relu(z)\bigr)}/{\partial z}$, is $0$ if $z < 0$, and $1$ if $z > 0$.
\begin{comment}
$$
\derp{\bigl(\relu(z)\bigr)}{z} =
\begin{cases}
0 & \textup{ if } z < 0, \\
1 & \textup{ if } z > 0.
\end{cases}
$$
\end{comment}
Moreover, $\alpha\bigl(W\bigr) = \alpha\bigl(I_n-\Phi^\top\Phi\bigr) = 1$, being $\Phi^\top\Phi \succeq 0$.
The result then follows by applying Lemma~\ref{thm:contractivity_fnn}.
\end{proof}

\smallskip

Essentially, with the above results we established that the trajectories of the \fcn are bounded. Next, we further characterize the stability of the equilibria of the \pfcn when the dictionary is RIP. We prove that the equilibrium is not only locally exponentially stable but also locally-strongly contracting in a suitably defined norm (see Figure \ref{fig:metatheorem}). 
\begin{comment}
In order to state and prove our result, it is convenient to introduce the following matrix:
\begin{equation}\label{eqn:S_epsilon}
 \qs_\eps :=
\begin{bmatrix}
\eps I_{\na} & 0 \\
0 & \eps^{-1} I_{n - \na}
\end{bmatrix}.   
\end{equation}
\end{comment}

%This latter result is a technical statement needed for accurately characterizing the linear-exponential convergence behavior of the \pfcn.

\newcommand{\strongweight}{\qs_{\eps}}
\bt[Local exponential stability and local strong contractivity of the \pfcn]
\label{thm:loc_exp_stab_loc_contractivity}
Let $\xstar\in \R_{\geq 0}^n \setminus \nondiffset{f}$ be an equilibrium point of the \pfcn~\eqref{eq:pfr-lca_x_dot} having $\na$ active neurons. If the dictionary $\Phi$ is \textup{RIP} of order $\na$ and parameter $\delta \in [0,1[$, then
\begin{enumerate}
\item $\xstar$ is locally exponentially stable; \label{item:locally-exp-stability_pfrlcn}
\item the \pfcn~\eqref{eq:pfr-lca_x_dot} is strongly infinitesimally contracting with rate $\ce > 0$ with respect to the norm $\norm{\cdot}_{2,\strongweight}$ in a neighborhood of $\xstar$.
\label{item:strongly_contracting_pfrlcn}
\end{enumerate}
\et
\begin{proof}
To prove item~\ref{item:locally-exp-stability_pfrlcn}, we show that $D\fpfcn(\xstar)$ is a Hurwitz matrix, i.e., $\alpha(D\fpfcn(\xstar))<0$.
We start noticing that
\begin{align}
    D\fpfcn(\xstar) = -I_n + \diag{d}\bigl(I_n-\trasp{\Phi}\Phi\bigr),
\end{align}
where $[d] = \left[ \partial \relu\bigl((I_n - \trasp{\Phi}\Phi)\xstar + \trasp{\Phi}u - \lambda \1_n\bigr)\right]$ is a diagonal matrix having diagonal entries equal to $0$ or $1$.
We let $\subscr{n}{a}$ and $\subscr{n}{ia}$ be the number of active and inactive neurons of $\xstar$, respectively, and rearrange the ordering of the elements in $\xstar$ such that $\xstar = [\subscr{\xstar}{a}, \quad {\subscr{\xstar}{ia}}]^\top$, where, ${\subscr{\xstar}{a}} \in \R^{\na}$ and ${\subscr{\xstar}{ia}} \in \R^{\nia}$, so that
\begin{align}
    \diag{d} = 
    \begin{bmatrix}
        I_{\na} & 0 \\ 0 & 0
    \end{bmatrix}.
\end{align}
Further, we also decompose $\trasp{\Phi}\Phi$ into
\begin{align}
\displaystyle \trasp{\Phi}\Phi = 
    \begin{bmatrix}
        \subscr{\trasp{\Phi}}{a}\subscr{\Phi}{a} & \subscr{\trasp{\Phi}}{a}\subscr{\Phi}{ia} \\ \subscr{\trasp{\Phi}}{ia}\subscr{\Phi}{a} & \subscr{\trasp{\Phi}}{ia}\subscr{\Phi}{ia}
    \end{bmatrix},
\end{align}
where $\subscr{\trasp{\Phi}}{a}\subscr{\Phi}{a} \in \R^{\na\times \na}$, $\subscr{\trasp{\Phi}}{a}\subscr{\Phi}{ia}\in \R^{\na\times \nia}$, $\subscr{\trasp{\Phi}}{ia}\subscr{\Phi}{a} \in \R^{\nia\times \na}$, $\subscr{\trasp{\Phi}}{ia}\subscr{\Phi}{ia}\in \R^{\nia\times \nia}$.\\ 
The fact that $\Phi$ is RIP of order $\na$ implies that
\begin{align*}
    \norm{\subscr{\trasp{\Phi}}{a} \subscr{x}{a}}_2^2 &=\trasp{x_a} \subscr{\trasp{\Phi}}{a}\subscr{\Phi}{a} x_a = 
    \trasp{\begin{bmatrix}
        x_a \\ \0_{\nia}
    \end{bmatrix}}
    \begin{bmatrix}
        \subscr{\trasp{\Phi}}{a}\subscr{\Phi}{a} & \subscr{\trasp{\Phi}}{a}\subscr{\Phi}{ia} \\ \subscr{\trasp{\Phi}}{ia}\subscr{\Phi}{a} & \subscr{\trasp{\Phi}}{ia}\subscr{\Phi}{ia}
    \end{bmatrix}
    \begin{bmatrix}
        \subscr{x}{a} \\ \0_{\nia}
    \end{bmatrix}
    \geq (1 -  \delta) \norm{\subscr{x}{a}}_2^2 > 0.
\end{align*}
Therefore, $\subscr{\trasp{\Phi}}{a}\subscr{\Phi}{a}$ is positive definite, and its smallest eigenvalue is bounded below by $1-\delta$. Moreover, $D\fpfcn(\xstar)$ can be written as
\begin{align}
\label{eq:DflcaPFR(xstar)}
D\fpfcn(\xstar) &=
\begin{bmatrix}
- \subscr{\trasp{\Phi}}{a}\subscr{\Phi}{a} & -\subscr{\trasp{\Phi}}{a}\subscr{\Phi}{ia} \\
0 & -I_{\nia}
\end{bmatrix},
\end{align}
that is a block upper triangular matrix with Hurwitz diagonal block matrices, and $\alpha{\bigl(D\fpfcn(\xstar)\bigr)} = \delta - 1 < 0$. Thus $D\fpfcn(\xstar)$ is Hurwitz. This concludes the proof.

Next, to prove~\ref{item:strongly_contracting_pfrlcn} we note that $\subscr{\lambda}{min}\bigl(\subscr{\trasp{\Phi}}{a}\subscr{\Phi}{a}\bigr) \geq 1-\delta$. By applying Corollary~\ref{cor:mu2_bound-for-Df} in Appendix~\ref{axp:l2_log_norm_of_upper_triangular_block_matrices} to $D\fpfcn(\xstar)$, we have $$\wlognorm{2}{\strongweight}{D\fpfcn(\xstar)} \leq - \ce < 0,$$ with the explicit expression of $\ce$ and $\qs_\eps$ given in~\eqref{eq:c_exp}.
% and with $\qs_\eps$ being defined in~\eqref{eqn:S_epsilon}.
Let $\mathcal{K}$ be the region of differentiable points in a neighborhood of $\xstar$. Then, by the continuity property of the log-norm, there exists a neighborhood of $\xstar$,
\beq
\label{eq:ball:strong_contractivy}
\ball{\strongweight}{\radiusstrong} := \setdef{z\in\R^n}{\norm{z - \xstar}_{2,\strongweight}\leq \radiusstrong}, \text{ with } \radiusstrong := \sup \setdef{z > 0}{\ball{\strongweight}{z} \subset \mathcal{K}},
\eeq
%a ball $\ball{\strongweight}{{\radiusstrong}}$, centered at $\xstar$ with radius $\radiusstrong := \sup \setdef{z > 0}{\ball{\strongweight}{z} \subset \mathcal{K}}$,
where $D\fpfcn(x)$ exists and $\wlognorm{2}{\qs_\eps}{D\fpfcn(x)} \leq - \ce$, for all $x \in \ball{\strongweight}{{\radiusstrong}}$.
This concludes the proof.
\end{proof}

\smallskip

\begin{rem}
To improve readability, in the statement of  Theorem~\ref{thm:loc_exp_stab_loc_contractivity} we do not provide the explicit expression for $\ce$ and $\qs_\eps$. These are instead given in~\eqref{eq:c_exp} of Appendix~\ref{axp:l2_log_norm_of_upper_triangular_block_matrices}. For the same reason, we do not report in the statement of  Theorem~\ref{thm:loc_exp_stab_loc_contractivity} the neighborhood in which the \pfcn is strongly infinitesimally contracting. However, as apparent from the proof, the neighborhood is $\ball{\strongweight}{\radiusstrong}$, which is defined in~\eqref{eq:ball:strong_contractivy}.
\end{rem}
%\grtodo{A reviewer could say: then give me also the neighborhood for the local exp. stability. We might decide to remove the second part of Remark 3.}
\smallskip

\begin{comment}
In Theorem~\ref{thm:loc_exp_stab_loc_contractivity} we require that $\xstar \in \R^n_{\geq 0} \setminus \Omega_f$. This assumption is not restrictive and can be always satisfied by picking $\lambda$ such that $u \neq \lambda \Phi\1_n$. In fact, for the \pfcn~\eqref{eq:pfr-lca_x_dot} the set of points in which $\fpfcn$ is not differentiable is $\Omega_f = \setdef{x \in \R^n}{(I_n-\trasp{\Phi}\Phi)x + \trasp{\Phi} u - \lambda \1_n = \0_n}$. Being $\xstar$ an equilibrium point of~\eqref{eq:pfr-lca_x_dot}, it is an optimal solution of the positive SR problem~\eqref{eq:positive_E_lasso_unconstrained}, thus $u \approx \Phi \xstar$. We compute
\begin{align}
(I_n-\trasp{\Phi}\Phi)x + \trasp{\Phi} u - \lambda \1_n &\approx \xstar - \trasp{\Phi}\Phi\xstar + \trasp{\Phi}\Phi \xstar  - \lambda \1_n \\
&= \xstar  - \lambda \1_n.
\end{align}
Therefore we have that $\xstar \in \Omega_f$ if and only if $\xstar  - \lambda \1_n = \O_n$.
Finally, by multiplying both side for $\Phi$, we obtain $\Phi\xstar  - \lambda \Phi\1_n = u - \lambda \Phi\1_n= \O_m$.

%%% FROM V: I would not write it down now. Also, to write it formally maybe we need some \epsilon in the condition $u \neq \lambda \Phi\1_n$.
\end{comment}

\smallskip
With the next result %(Theorem~\ref{thm:GLin-ExpS_of_pfrlcn})
we prove that the \pfcn converges linear-exponentially to $\xstar$ (see Figure~\ref{fig:metatheorem}). The proof is given in the Methods section, where we prove a more general result (see Corollary~\ref{thm:fin_decay_fin_time_GWC_and_LSC_dynamics}, in Section~\ref{sec:proof_theorem}).
%see Section~\ref{sec:proof_theorem} at the end of the Methods section). 
We summarize the key symbols used in the next theorem in Table~\ref{table:symbols}.
%Theorem~\ref{thm:GLin-ExpS_of_pfrlcn} in Table~\ref{table:symbols}.
\begin{table}
\centering
\resizebox{\textwidth}{!}{%
\begin{tabular}{|c|l|l|}
\hline
\multicolumn{1}{|c|}{Symbol} & \multicolumn{1}{c|}{Meaning} & \multicolumn{1}{c|}{Ref.}\\
\hline
$\qw$ & Weight matrix w.r.t. the \pfcn is globally-weakly contracting & Equation~\eqref{eq:weight_matrix_wc}\\
\hline
$\norm{\cdot}_{2,\qw}$ & Euclidean weighted norm w.r.t. the \pfcn is globally-weakly contracting & Theorem~\ref{thm:weak-contractivity_pfrlcn} \\
\hline
$\strongweight$ & Weight matrix w.r.t. the \pfcn is locally-strongly contracting & Equation~\eqref{eq:c_exp}\\
\hline
$\norm{\cdot}_{2,\strongweight}$ & Euclidean weighted norm w.r.t. the \pfcn is locally-strongly contracting & Theorem~\ref{thm:loc_exp_stab_loc_contractivity} \\
\hline
$\prodeqnorm_{\strongweight,\qw}$ & Equivalence ratio between $\norm{\cdot}_{2,\strongweight}$ and $\norm{\cdot}_{2,\qw}$ & Definition~\ref{def:equivalence_ratio} \\
\hline
$\radiusstrong$ & Radius of the ball where the system is strongly infinitesimally contracting & Equation~\eqref{eq:ball:strong_contractivy} \\
\hline
$\ball{\strongweight}{{\radiusstrong}}$ & Ball of radius $\radiusstrong$ centered at $\xstar$ computed w.r.t. $\norm{\cdot}_{2,\strongweight}$ %Ball around $\xstar$ where the system is strongly infinitesimally contracting
& Equation~\eqref{eq:ball:strong_contractivy} \\ 
\hline
$\radius$ & Radius of the largest ball $\ball{\qw}{\radius}$ contained in $\ball{\strongweight}{{\radiusstrong}}$ & Theorem~\ref{thm:GLin-ExpS_of_pfrlcn}\\
\hline
$\ball{\qw}{\radius}$ & Ball of radius $\radius$ centered at $\xstar$ computed w.r.t. $\norm{\cdot}_{2,\qw}$ & Theorem~\ref{thm:GLin-ExpS_of_pfrlcn}\\ 
\hline
$\ce$ & Exponential decay rate & Equation~\eqref{eq:c_exp}\\
\hline
$\cl$ & Average linear decay rate & Equation~\eqref{eq:c_lin_pfcn}\\
\hline
$\tld$ & Linear-exponential crossing time & Equation~\eqref{eq:t_cross_pfcn}\\
\hline
$\rho$ & Contraction factor, $0 < \rho < 1$ & Theorem~\ref{thm:GLin-ExpS_of_pfrlcn}\\
\hline
\end{tabular}
}
\caption{Symbols used in the linear-exponential bound~\eqref{eq:lin_exp_bound}.}
\label{table:symbols}
\end{table}

\smallskip
\bt[Linear-exponential stability of the \pfcn]
\label{thm:GLin-ExpS_of_pfrlcn}
Consider the \pfcn~\eqref{eq:pfr-lca_x_dot}
%with dictionary $\Phi\in \R^{m\times n}$ and parameter $\lambda>0$.
under the same assumptions and notations of Theorems~\ref{thm:weak-contractivity_pfrlcn} and~\ref{thm:loc_exp_stab_loc_contractivity}. 
Let $\ball{\strongweight}{{\radiusstrong}}$ be the ball around $\xstar$ where the system is strongly infinitesimally contracting. Then, for each trajectory $x(t)$ starting from $x(0) \notin \ball{\strongweight}{{\radiusstrong}}$ and for any $0 < \rho < 1$, the distance $\norm{x(t) - \xstar}_{2,\qw}$ decreases \emph{linear-exponentially}, in the sense that:
%let $\radius > 0$ be such that $\ball{\qw}{\radius} \subset \ball{\strongweight}{{\radiusstrong}}$. Then, for each trajectory $x(t)$ starting from $x(0) \notin \ball{\strongweight}{{\radiusstrong}}$ and for any contraction factor $0 < \rho < 1$, the distance $\norm{x(t) - \xstar}_{2, \qw}$ decreases \emph{linear-exponentially with time}, in the sense that:
\begin{align}
\norm{x(t) - \xstar}_{2, \qw} 
&\leq 
\begin{cases}
\norm{x(0) - \xstar}_{2, \qw} +(1 - \rho)\radius - \cl t ,
\qquad & \text{if } t \leq \tld,\\
\, \prodeqnorm_{\strongweight,\qw} \, \radius \, \e^{ - \ce (t - \tld)}
\quad & \text{if } t>\tld,
\end{cases}
\label{eq:lin_exp_bound}
\end{align}
where $\radius > 0$ is the radius of the largest ball $\ball{\qw}{\radius}$ centered at $\xstar$ such that $\ball{\qw}{\radius} \subset \ball{\strongweight}{{\radiusstrong}}$ and where
\begin{align}
\label{eq:c_lin_pfcn}
\cl &= \frac{\ce(1- \rho)\radius}{\ln(\prodeqnorm_{\strongweight,\qw} \rho^{-1})}, \\
\tld &= \Bigceil{\frac{\norm{x(0)-\xstar}_{2, \qw} - \radius}{(1-\rho)\radius}} \frac{\ln(\prodeqnorm_{\strongweight,\qw} \rho^{-1})}{\ce},
\label{eq:t_cross_pfcn}
\end{align}
are the \emph{average linear decay rate} and the \emph{linear-exponential crossing time}, respectively.
\et
\smallskip
In what follows, we simply term $\rho$ as \emph{contraction factor}. We also note that the contraction factor can be chosen optimally to maximize the average linear decay rate $\cl$. We refer to Lemma~\ref{lem:optimal_contraction_factor} for the mathematical details and the exact statement.
\section{Simulations}
We now illustrate the effectiveness of the \pfcn in solving the positive SR problem~\eqref{eq:positive_E_lasso_unconstraines} via a numerical example\footnote{The code to replicate all the simulations in this section is available at the github \url{https://tinyurl.com/PFCN-for-Sparse-Reconstruction}.} that is built upon the one in~\cite{AB-JR-CJR:12}. To this aim, we consider a $n = 512$ dimensional sparse signal $y_0 \in \R_{\geq 0}^n$, with $\na = 5$ randomly selected non-zero entries. The amplitude of these non-zero entries is obtained by drawing from a normal Gaussian distribution and then taking the absolute values. As in~\cite{AB-JR-CJR:12} the dictionary $\Phi \in \R^{m \times n}$ is built as a union of the canonical basis and a sinusoidal basis (each basis is normalized so that the dictionary columns have unit norms). Also, we set: (i) the measurements $u \in \R^m$, with $m = 256$, to be $u = \Phi y_0 + \eta$, where $\eta$ is a Gaussian random noise with standard deviation $\sigma = 0.0062$; (ii)  $\lambda = 0.025$.

Given this set-up, we simulated both the \pfcn~\eqref{eq:pfr-lca_x_dot} and, for comparison, the LCA~\eqref{eq:lca-soft-xi}. Simulations were performed  with Python using the ODE solver \emph{solve\_ivp}. In all the numerical experiments, the simulation time was $t \in [0, 15]$ and initial conditions were set to $0$, except for $20$ randomly selected neurons (initial conditions were kept constant across the simulations). The time evolution of the state variables for both the \pfcn and LCA is shown in Figure~\ref{fig:sim-trajectories_pfcn_soft}. Both panels illustrate that both the \pfcn and the LCA converge to an equilibrium that is close to $y_0$ (although it can not be exactly $y_0$ because of the measurement noise). Also, the figure clearly shows, in accordance with Lemma~\ref{lem:positive_frcn}, that the trajectories of the \pfcn are always non-negative. Instead, the trajectories of the LCA nodes exhibit also negative values over time.

\begin{figure}[!h]
\centering 
\includegraphics[width=.82\linewidth]{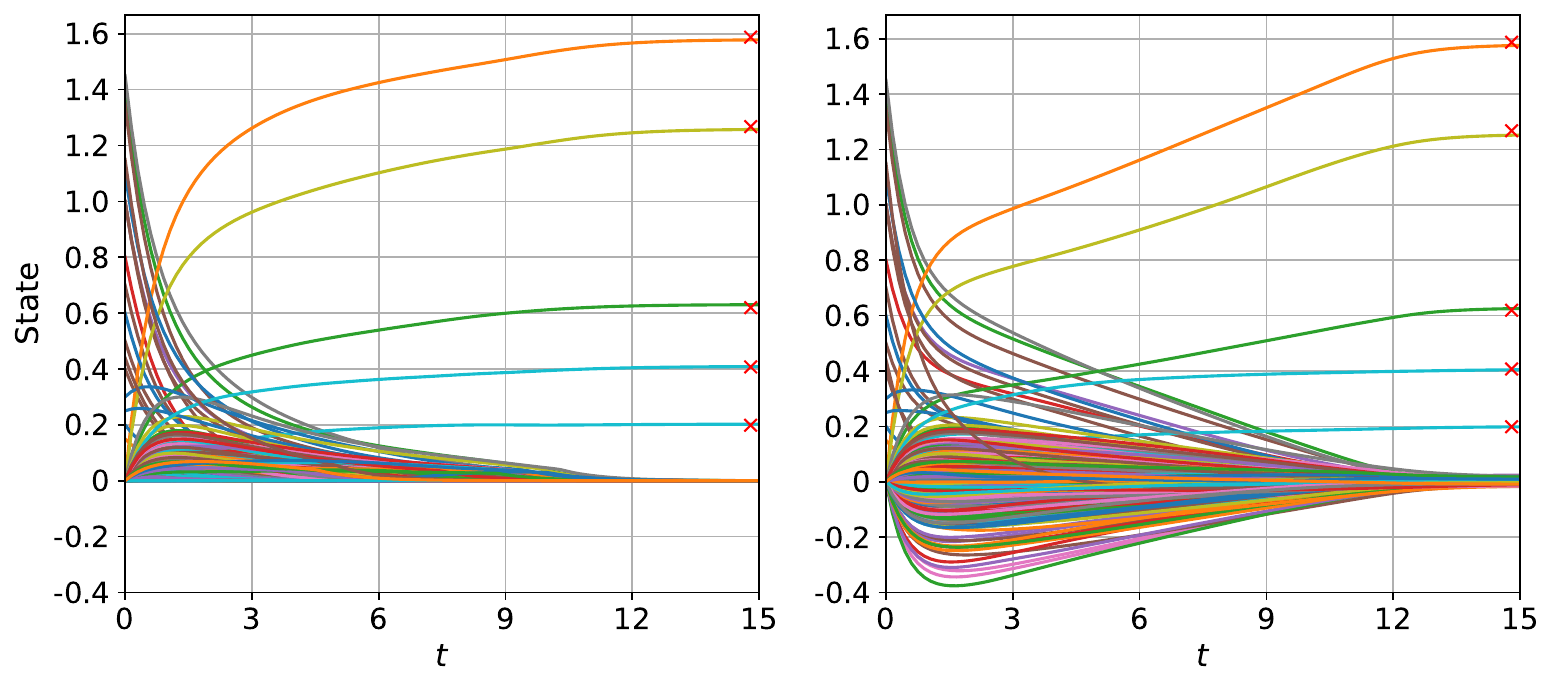}
\caption{Time evolution of the state/neuron variables of the proposed \pfcn~\eqref{eq:pfr-lca_x_dot} (leftward panel) and of the LCA~\eqref{eq:lca-soft-xi} (rightward panel) networks. The cross symbols are the non-zero elements of the sparse vector $y_0$. Both the \pfcn and the LCA converge to an equilibrium that is close to $y_0$. Note that, in accordance with Lemma~\ref{lem:positive_frcn}, the state variables of the \pfcn never become negative.}
\label{fig:sim-trajectories_pfcn_soft}
\end{figure}

To illustrate the global convergence behavior of the \pfcn~\eqref{eq:pfr-lca_x_dot}, we performed an additional set of simulations, this time with the \pfcn starting  from $20$ randomly generated initial conditions. Then, we randomly selected two neurons from the active and inactive sets and recorded their evolution. The result of this process is shown in Figure~\ref{fig:sim-multiple_trajectories}, which reports a projection of the phase plane defined by these nodes. Figure~\ref{fig:sim-multiple_trajectories} shows that the trajectories of the selected nodes converge to the equilibrium point from any of the chosen initial conditions. Specifically, in accordance with our results, the trajectories of the active neurons converge to positive values, while the trajectories of the inactive nodes converge to the origin.

\begin{figure}[!h]
\centering 
\includegraphics[width=.82\linewidth]{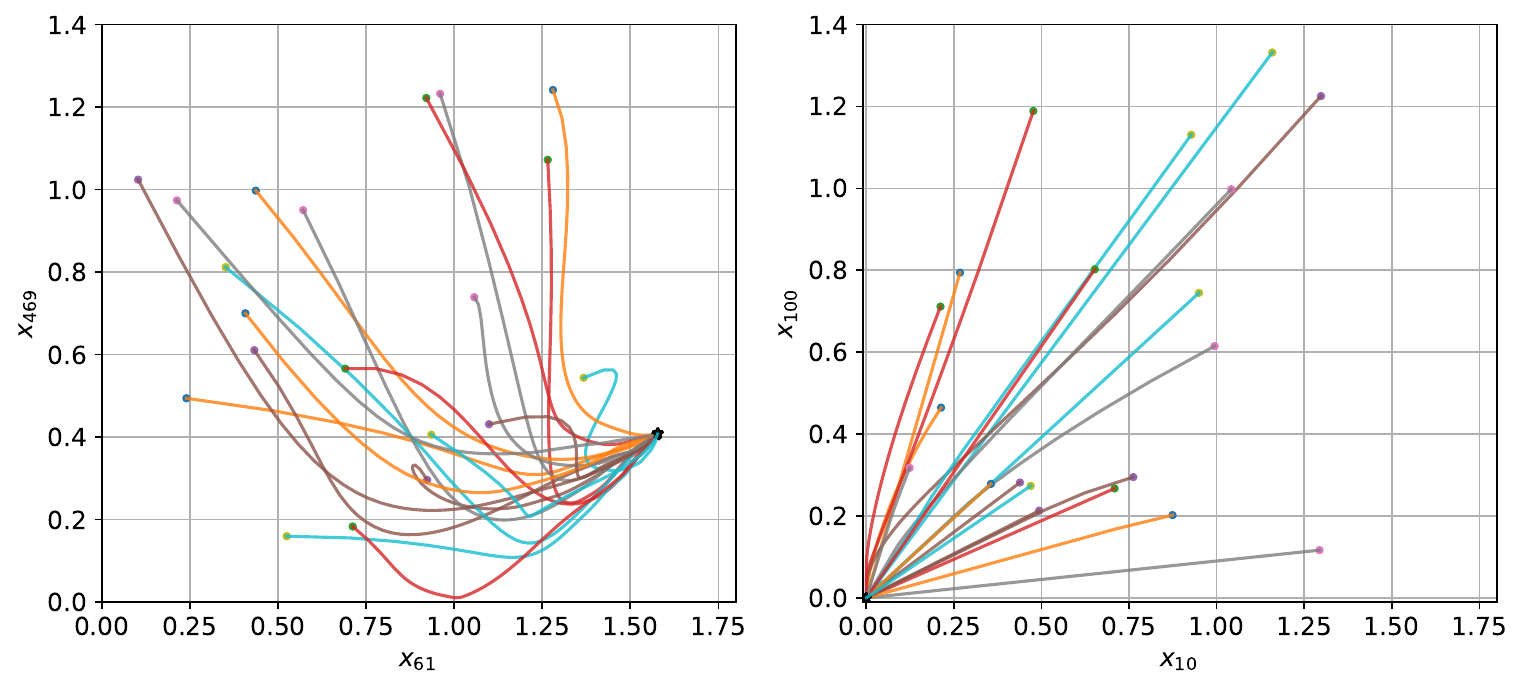}
\hspace{.05\linewidth}
\caption{Trajectories of two randomly chosen nodes of the \pfcn~\eqref{eq:pfr-lca_x_dot} from the active (leftward panel) and inactive (rightward panel) set in the planes defined by these two nodes, respectively. In the panels, the evolution is shown from $20$ randomly chosen initial conditions. In accordance with our results, the trajectories of the active neurons converge to positive values, while the trajectories of the inactive nodes converge to the origin.}
\label{fig:sim-multiple_trajectories}
\end{figure}

Finally, we performed an additional, exploratory, numerical study to investigate what happens when the activation function of the LCA is the shifted $\relu$. Even though the assumptions in~\cite{AB-JR-CJR:12} on the activation function exclude the use of the $\relu$ for the LCA, we decided to simulate this scenario to investigate if the LCA dynamics would become positive if the $\relu$ was used as activation function. Hence, for our last numerical study, we considered the following LCA dynamics:

\beq
\dot x(t) = - x(t) + \bigl(I_n + \Phi^\top \Phi\bigr)\relu\bigl(x(t)- \lambda \1_n\bigr) + \Phi^\top u(t),
\label{eq:lca-relu-x}
\eeq
with output $y(t) = \relu\bigl(x(t)- \lambda \1_n\bigr)$. In Figure~\ref{fig:sim-trajectories_pfcn_relu}  the time evolution of the state variables of the LCA~\eqref{eq:lca-relu-x} is shown. As apparent from the figure, even using the $\relu$ as activation function, the trajectories of the LCA states still exhibit  negative values over time.

\begin{SCfigure}[10][h!]
\centering
\includegraphics[width=0.5\linewidth]{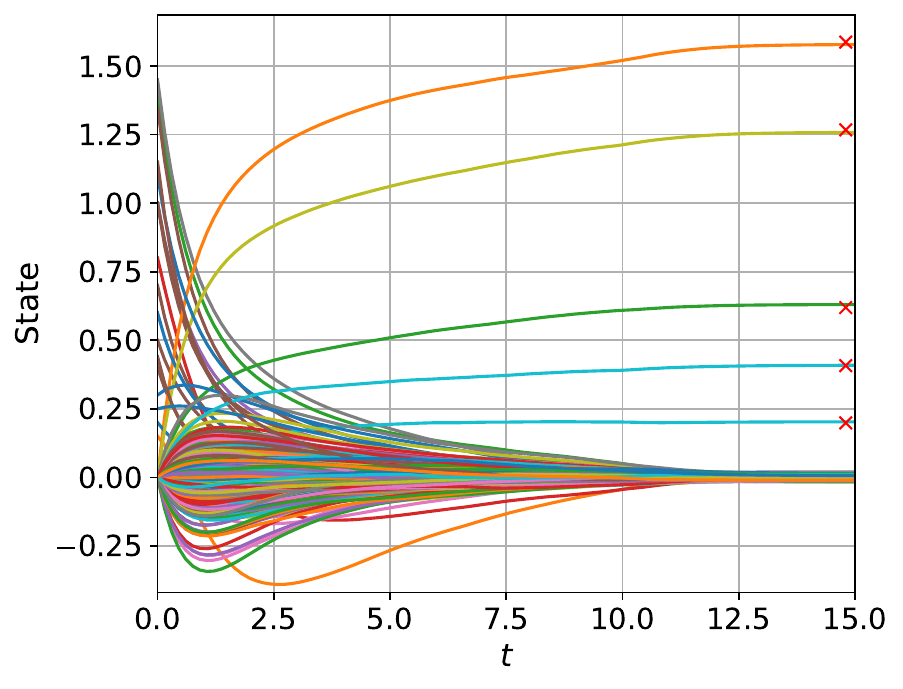}
\caption{Time evolution of the state variables of the LCA~\eqref{eq:lca-relu-x} with $\relu$ as activation function. The cross symbols are the non-zero elements of $y_0$. The LCA converges to an equilibrium close to $y_0$. Even using the $\relu$ as activation function, the trajectories of the LCA states still exhibit negative values over time.}
\label{fig:sim-trajectories_pfcn_relu}
\end{SCfigure}

\section{Methods}
\label{sec:methods}
We prove a number of results, from which Theorem~\ref{thm:GLin-ExpS_of_pfrlcn} follows, that  characterize convergence of nonlinear systems of the form of~\eqref{eq:dynamical_system} that are globally-weakly contracting and locally-strongly contracting (possibly, in different norms). We show that these systems, which also arise in flow dynamics, traffic networks, and primal-dual dynamics~\cite{SJ-PCV-FB:19q, FB:23-CTDS}, have a linear-exponential convergence behavior.

First, we give a general algebraic result on the inclusion relationship between balls computed with respect to different norms.
\smallskip

\begin{lem}[Inclusion between balls computed with respect to different norms]
\label{lem:ineq_balls_inclusions}
Given two norms $\norm{\cdot}_{\alpha}$ and $\norm{\cdot}_{\beta}$ on $\R^n$, for all $r>0$, it holds that
\begin{equation}
\label{ineq:balls:inclusions}
\ball{\beta}{r/\kab} \subseteq \ball{\alpha}{r} \subseteq  \ball{\beta}{r \kba},
\end{equation}
with $\kab$ and $\kba$ given in~\eqref{eq:equivalence_coeff_norms}.
\end{lem}
\begin{proof}
We start by proving the inequality $\ball{\alpha}{r} \subseteq \ball{\beta}{r \kba}$.
By definition of ball of radius $r$, for any $x\in \ball{\alpha}{r}$, we know that $\norm{x}_{\alpha}\leq{r}$. Also, we have
$$\frac{1}{\kba}\norm{x}_{\beta}\leq \frac{1}{\kba}\kba\norm{x}_{\alpha}\leq\norm{x}_{\alpha}\leq{r}.$$
Therefore $\norm{x}_{\beta}\leq{r}\kba$, so that $x\in \ball{\beta}{r\kba}$.

The inequality $\ball{\beta}{r/\kab} \subseteq \ball{\alpha}{r}$ follows directly from the above inequality and from the fact that $\kab \kba \geq 1$. Specifically, we have:
$$\norm{x}_{\alpha}\leq \kab\norm{x}_{\beta} \leq \kab r \kba \leq r.$$
\end{proof}

\smallskip

We are now ready to state the main result of this section.
\smallskip
\bt[Finite decay in finite time of globally-weakly and locally-strongly contracting systems]
\label{thm:fin_decay_fin_time_GWC_and_LSC_dynamics}
Let $\norm{\cdot}_{\loc}$ and $\norm{\cdot}_{\glo}$ be two norms on $\R^n$. Consider system~\eqref{eq:dynamical_system} with $\map{f}{\R_{\geq 0} \times \R^n}{\R^n}$ being a locally Lipschitz map satisfying the following assumptions
\begin{enumerate}[label=\textup{($A$\arabic*)}, leftmargin=1.4 cm,noitemsep]
\item \label{ass:globally-weakly_contractivity}
$f$ is weakly infinitesimally contracting on $\R^n$ with respect to $\norm{\cdot}_\glo$;
\item \label{ass:locally-strongly_contractivity}
$f$ is $\ce$-strongly infinitesimally contracting on a forward-invariant set $\mcS$
with respect to $\norm{\cdot}_\loc$;
\item \label{ass:eq_point}
$\xstar \in \mcS$ is an equilibrium point, i.e., $f(t,\xstar) = \0_n$, for all $t \geq 0$.
\end{enumerate}
Also, let $\ball{\glo}{r} \subset \mcS$ be the largest closed ball centered at $\xstar$ with radius $\radius>0$ with respect to  $\norm{\cdot}_{\glo}$.
Then, for each trajectory $x(t)$ starting from $x(0) \notin \mcS$ and for any {contraction factor} $0 < \rho < 1$, the distance along the trajectory decreases at worst linearly with an \emph{average linear decay rate}
\begin{equation}
\ds \cl = \frac{(1-\rho) \radius}{t_\rho},
\end{equation}
up to at most the \emph{linear-exponential crossing time}
 \begin{equation}
    \tld = \Bigceil{\frac{\glbnorm{x(0)-\xstar} - \radius}{(1-\rho)\radius}} t_\rho,
\end{equation}
when the trajectory enters $\ball{\glo}{r}$.
%\begin{enumerate}
%    \item for each trajectory $x(t)$ starting from $x(0) \notin \mcS$ and for any {contraction factor} $0 < \rho < 1$:
%\begin{equation}
%    \norm{x(t_\rho) - \xstar}_\glo \leq
%    \norm{x(0) - \xstar}_\glo - (1-\rho)\radius,
%\end{equation}
%where $t_\rho={\ln(\ratiolg \rho^{-1})}/{\ce}$;
%\item the distance along the trajectory decreases with an \emph{average linear decay rate}
%\begin{equation}
%\ds \cl = \frac{(1-\rho) \radius}{t_\rho},
%\end{equation}
%and the trajectory enters $\ball{\glo}{r}$ in at most 
%a \emph{linear-exponential crossing time}
%\begin{equation}
%    \tld = \Bigceil{\frac{\glbnorm{x(0)-\xstar} - \radius}{(1-\rho)\radius}} t_\rho.
%\end{equation}
%\end{enumerate} 
\et
\bigskip

\begin{proof}
Consider a trajectory $x(t)$ of~\eqref{eq:dynamical_system} starting from initial condition $x(0)\not\in\mcS$ and define the \emph{intermediate point} $\xtmp = \xstar+\radius \frac{x(0)-\xstar}{\glbnorm{x(0)-\xstar}{}}$, as in Figure~\ref{fig:contractivity-weak-localglobal-new}. Note that $\xtmp$ is a point on the boundary of $\ball{\glo}{\radius}$, since $\norm{\xtmp - \xstar}_\glo = \radius$.
Moreover, the points $\xstar$, $\xtmp$, and $x(0)$ lie on the same line segment, thus
\beq
\norm{x(0) - \xstar}_\glo = \norm{x(0) -\xtmp}_\glo + \norm{\xtmp - \xstar}_\glo.
\label{eq:aligned_points}
\eeq

\begin{SCfigure}[9][!ht]
\centering
\includegraphics[width=.52\linewidth]{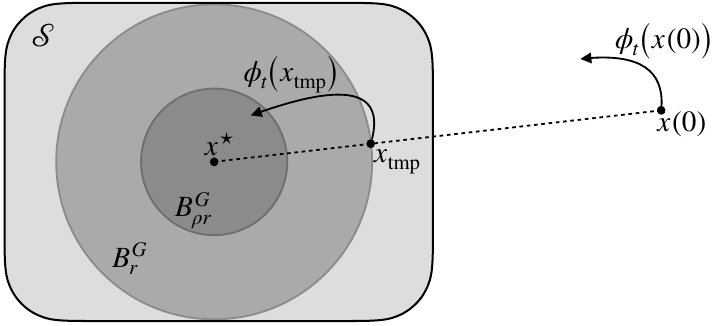}
\hspace{.01\linewidth}
\caption{Illustration of the set up for the proof of Th.~\ref{thm:fin_decay_fin_time_GWC_and_LSC_dynamics} with $\glo = \norm{\cdot}_{2}$. Given the equilibrium point $\xstar \in \mcS$, with $\mcS$ forward invariant set, we consider a trajectory $\odeflowtx{t}{x(0)}$ of~\eqref{eq:dynamical_system} starting from $x(0)\not\in\mcS$ and define the \emph{intermediate point} $\xtmp \in \ball{\glo}{\radius}$. After a time $t_\rho$ the trajectory starting at $\xtmp$ (which may exit $\ball{\glo}{\radius}$)  must enter $\ball{\glo}{\rho\radius}$, for $0<\rho<1$. Image reused with permission from~\cite{FB:23-CTDS}.}
\label{fig:contractivity-weak-localglobal-new}
\end{SCfigure}

Using the triangle inequality, we get
\begin{align*}
\glbnorm{ \odeflowtx{t}{x(0)} -\xstar} &\leq \glbnorm{ \odeflowtx{t}{x(0)} - \odeflowtx{t}{\xtmp}} + \glbnorm{ \odeflowtx{t}{\xtmp} - \xstar}.
\end{align*}
By Assumption~\ref{ass:globally-weakly_contractivity} and equality~\eqref{eq:aligned_points}, we know that $\glbnorm{\odeflowtx{t}{x(0)} - \odeflowtx{t}{\xtmp}} \leq \glbnorm{x(0) - \xtmp} = \glbnorm{x(0) - \xstar}-\radius$, thus
\begin{align*}
\glbnorm{ \odeflowtx{t}{x(0)} -\xstar} &\leq \glbnorm{x(0)-\xstar}-\radius + \glbnorm{ \odeflowtx{t}{\xtmp} -\xstar}.
\end{align*}
Next, we upper bound the term $\glbnorm{ \odeflowtx{t}{\xtmp} -\xstar}$. 
We note that, because each trajectory originating in $\ball{\glo}{\radius}$ remains in $\mcS$, the time required for each trajectory starting in $\ball{\glo}{\radius}$, to be inside $\ball{\glo}{\rho \radius}$ for the $\ce$-strongly contracting map $f$ is 
\begin{equation*}
t_\rho  = \frac{\ln(\ratiolg \rho^{-1})}{\ce}.
\end{equation*}
This follows by noticing that  
\begin{align*}
x(0)\in \ball{\glo}{r} \quad&\overset{\eqref{ineq:balls:inclusions}, \text{ $\supscr{2}{nd}$ inequality}}{\implies}\quad x(0)\in \ball{\loc}{r k_{\loc}^{\glo}}, \\
x(t_\rho)\in \ball{\glo}{\rho r} \quad&\overset{\eqref{ineq:balls:inclusions}, \text{ $\supscr{1}{st}$ inequality}}{\impliedby}\quad
x(t_\rho)\in \ball{\loc}{\rho r/k_{\glo}^{\loc}}.
\end{align*}
Thus, the time required for a trajectory starting in $\ball{\glo}{r}$ to be inside $\ball{\glo}{\rho r}$ is upper bounded by the time required for the trajectory to go from $\ball{\loc}{r k_{\loc}^{\glo}}$ to $\ball{\loc}{\rho r/k_{\glo}^{\loc}}$. In these balls, Assumption~\ref{ass:locally-strongly_contractivity} implies $\norm{x(t)}_{\loc}\leq\e^{-\ce t}\norm{x(0)}_{\loc}$ and so $t_\rho$ is determined by the equality $\e^{-\ce t_\rho}r k_{\loc}^{\glo} = \rho r/k_{\glo}^{\loc}$.

Therefore, at time $t_\rho$, we know $\odeflowtx{t_\rho}{\xtmp} \in \ball{\glo}{\rho\radius}$ and we have
\begin{align*}
\glbnorm{\odeflowtx{t_\rho}{x(0)} -\xstar} &\leq \glbnorm{x(0)-\xstar}-\radius + \rho \radius =\glbnorm{x(0)-\xstar}-(1-\rho) \radius.
\end{align*}
By iterating the above argument, it follows that after each interval of duration $t_\rho$, the distance $\glbnorm{x(t)-\xstar}$ has decreased by an amount $(1-\rho)\radius$ for each $x(t)$.
Therefore the average linear decay satisfies
\begin{equation}
\label{eq:ave-lin-decay-rate}
\cl := \frac{\text{variation in distance to $\xstar$}}{\text{variation in time}} = \frac{(1-\rho) \radius}{t_\rho} = \ce\radius \frac{1-\rho}{\ln(\ratiolg \rho^{-1})}.
\end{equation}
Hence, after at most a linear-exponential crossing time $ \tld := \Bigceil{\frac{\glbnorm{x(0)-\xstar} - \radius}{(1-\rho)\radius}}{t_\rho}$, the trajectory will be inside $\ball{\glo}{r} \subset \mcS$.
\begin{comment}
Moreover, we note that Assumption~\ref{ass:locally-strongly_contractivity} implies the exponential decrease, i.e.,
\[
\norm{\odeflowtx{t}{x(0)} -\xstar}_{\loc} \leq  \norm{x(0) - \xstar}_{\loc}\,\e^{-\ce( t - \tld) }, \quad \forall t > \tld.
\]
Applying the equivalence of norms to the above inequality we have
\begin{align*}
\glbnorm{\odeflowtx{t}{x(0)} -\xstar} &\leq \prodeqnorm_{\loc, \glo} \, \radius \, \e^{-\ce( t - \tld)}, \quad \forall t > \tld.
\end{align*}
\end{comment}
This concludes the proof.
\end{proof}

\smallskip

\begin{rem}
Assumptions~\ref{ass:locally-strongly_contractivity} and~\ref{ass:eq_point} of Theorem~\ref{thm:fin_decay_fin_time_GWC_and_LSC_dynamics} imply that for any $x(0) \in \mcS$, the distance $\norm{x(t) - \xstar}_\loc$ decreases exponentially with time with rate $\ce$. Specifically, for all $t \geq 0$ it holds that
\beq
\label{eq:bound_inside_the_ball}
\norm{x(t) - \xstar}_\loc \leq \e^{-\ce t} \norm{x(0) - \xstar}_\loc.
\eeq
\end{rem}

%\begin{SCfigure}[10][!ht]
%\centering
%\includegraphics[width=.63\linewidth]{fig_paper/lin-exp-decay-bound_new.pdf}
%\hspace{.01\linewidth}
%\caption{Plot of the behavior of $\norm{x(t) - \xstar}_\glo$ as described in Theorem~\ref{thm:fin_decay_fin_time_GWC_and_LSC_dynamics} (dashed blue curve) and of the linear exponential bound given in equation~\eqref{eq:lin_exp_decay} (solid green curve) for $\radius = 1$, $\norm{x(0) - \xstar}_\glo = 4$, $\ce = $, $\prodeqnorm_{\loc, \glo} = 2$, and $\rho = \bar \rho$, with $\bar \rho$ being the optimal contractor factor given in equation~\eqref{eq:optimal_contraction_factor}.}
%\label{fig:bound}
%\end{SCfigure}

The next result, which establishes the linear-exponential convergence of system~\eqref{eq:dynamical_system}, follows from Theorem~\ref{thm:fin_decay_fin_time_GWC_and_LSC_dynamics}.

\smallskip
\begin{cor}[Linear-exponential decay of globally-weakly and locally-strongly contracting systems]
\label{cor:lin-exp-decay_GW-LS_contracting_systems}
Under the same assumptions and notations as in Theorem~\ref{thm:fin_decay_fin_time_GWC_and_LSC_dynamics}, for each $x(0) \notin \mcS$ and for any contraction factor $0 < \rho <1$, the distance $\norm{x(t) - \xstar}_\glo$ decreases linear-exponentially with time, in the sense that:
\begin{align}
\label{eq:lin_exp_decay}
\norm{x(t) - \xstar}_\glo 
&\leq 
\begin{cases}
\norm{x(0) - \xstar}_\glo +(1-\rho)r - \cl t , 
\qquad & \text{if }t \leq \tld ,  \\
\, \prodeqnorm_{\loc,\glo} \, r \, \e^{ - \ce (t - \tld)}
\quad & \text{if } t>\tld.
\end{cases}
\end{align}
\end{cor}
\begin{proof}
The result follows directly from Theorem~\ref{thm:fin_decay_fin_time_GWC_and_LSC_dynamics}. Indeed, given a trajectory $x(t)$ of~\eqref{eq:dynamical_system} starting from $x(0) \notin \mcS$, for all $t \leq \tld$, from Theorem~\ref{thm:fin_decay_fin_time_GWC_and_LSC_dynamics} we know that the distance $\norm{x(t) - \xstar}_\glo$ decreases linearly by an amount $(1-\rho)\radius$ with an average linear decay rate $\cl = {(1-\rho) \radius}/{t_\rho}$ towards $\ball{\glo}{r} \subset \mcS$, which implies the upper bound
$$
\norm{x(t) - \xstar}_\glo \leq \norm{x(0) - \xstar}_\glo + (1 - \rho)r - \cl t.
$$
Next, for all $t > \tld$ the trajectory $x(t)$ is inside $\ball{\glo}{r}$ and Assumption~\ref{ass:locally-strongly_contractivity}, i.e., $\ce$-strongly infinitesimally contractivity on $\mcS$, implies the bound
\[
\norm{\odeflowtx{t}{x(0)} -\xstar}_{\loc} \leq \norm{x(0) - \xstar}_{\loc}\,\e^{-\ce( t - \tld) }, \quad \forall t > \tld.
\]
Applying the equivalence of norms to the above inequality we have
\begin{align*}
\glbnorm{\odeflowtx{t}{x(0)} -\xstar} &\leq \prodeqnorm_{\loc, \glo} \, \norm{x(0) - \xstar}_{\glo} \, \e^{-\ce( t - \tld)}, \quad \forall t > \tld.
\end{align*}
Therefore, for all $t > \tld$ we have
\begin{align*}
\norm{x(t) - \xstar}_\glo := \glbnorm{\odeflowtx{t}{x(0)} -\xstar} \leq \prodeqnorm_{\loc, \glo} \, \radius \, \e^{-\ce( t - \tld)}.
\end{align*}
This concludes the proof.
\end{proof}
\smallskip
With the following Lemma, we give the explicit expression for the optimal contraction factor $\rho$ that maximizes the average linear decay rate $\cl$.

\smallskip

\begin{lem}[Optimal contraction factor]
\label{lem:optimal_contraction_factor}
Under the same assumptions and notations as in Theorem~\ref{thm:fin_decay_fin_time_GWC_and_LSC_dynamics}, for $\prodeqnorm_{\loc,\glo} > 1$ the contraction factor $\rho \in {]0,1[}$ that maximize the average linear decay rate $\cl$ is 
\beq
\label{eq:optimal_contraction_factor}
\bar \rho(\prodeqnorm_{\loc,\glo}) = - \frac{1}{W_{-1}(- \e^{-1} \prodeqnorm_{\loc,\glo}^{-1})},
\eeq
where $W_{-1}(\cdot)$ is the branch of the Lambert function $W(\cdot)$ \footnote{The Lambert function $W(\cdot)$ is a multivalued function defined by the branches of the converse relation of the function $f(x) = x\e^x$. See~\cite{RMC-GHG-DEGH-DJJ-DEK:96} for more details.}
satisfying $W(x) \leq -1$, for all $x \in [-1/\e,0[$.
\end{lem}
\begin{proof}
To maximize the linear decay rate $\cl$ we need to solve the optimization problem
\begin{align}
\max_{0<\rho<1} \frac{1-\rho}{\ln(\prodeqnorm_{\loc,\glo}) - \ln(\rho)}.
\end{align}
We compute
\begin{align*}
\frac{d}{d\rho} \frac{1-\rho}{\ln(\prodeqnorm_{\loc,\glo}) - \ln(\rho)} &= \frac{\rho \ln(\rho) - \rho\bigl(1+\ln(\prodeqnorm_{\loc,\glo})\bigr) + 1}{\rho\bigl(\ln(\prodeqnorm_{\loc,\glo}) - \ln(\rho)\bigr)^2} = 0,
\end{align*}
which holds if and only if
\begin{align}
\label{eq:derivative_rho}
\rho \ln(\rho) - \rho\bigl(1+\ln(\prodeqnorm_{\loc,\glo})\bigr) + 1 = 0.
\end{align}
Note that the equality~\eqref{eq:derivative_rho} is a transcendental equation of the form $x \ln(x) + ax + b = 0$, whose solution is known to be the value $x = \frac{- b }{ W_{0}(- b \e^{a})}$ if $- b \e^{a} \geq 0$ and the two values $x = \frac{-b}{ W_{0}(- b \e^{a})}$ and $x = \frac{- b }{ W_{-1}(- b \e^{a})}$ if $-1/\e \leq - b \e^{a} < 0$, where $W_0(\cdot)$ is the branch satisfying $W(x) \geq -1$, and $W_{-1}(\cdot)$ is the branch satisfying $W(x) \leq -1$.

In our case it is $b = 1$ and $a = -(1+\ln(\prodeqnorm_{\loc,\glo}))$, thus $ - b \e^{a}  =  - \e^{-(1+\ln(\prodeqnorm_{\loc,\glo}))} \in ]-\frac{1}{\e}, 0[$. Therefore, the solutions of the equality~\eqref{eq:derivative_rho} are $\rho = - \frac{1}{ W_{0}(- \e^{-1} \prodeqnorm_{\loc,\glo}^{-1})}$ and $\rho = -\frac{1}{ W_{-1}(- \e^{-1} \prodeqnorm_{\loc,\glo}^{-1})}$. Being $0 < \rho < 1$, the only admissible solution is $\rho = -\frac{1}{ W_{-1}(- \e^{-1} \prodeqnorm_{\loc,\glo}^{-1})}$, thus the thesis.
\end{proof}

\begin{SCfigure}[10][!ht]
\centering
\includegraphics[width=.6\linewidth]{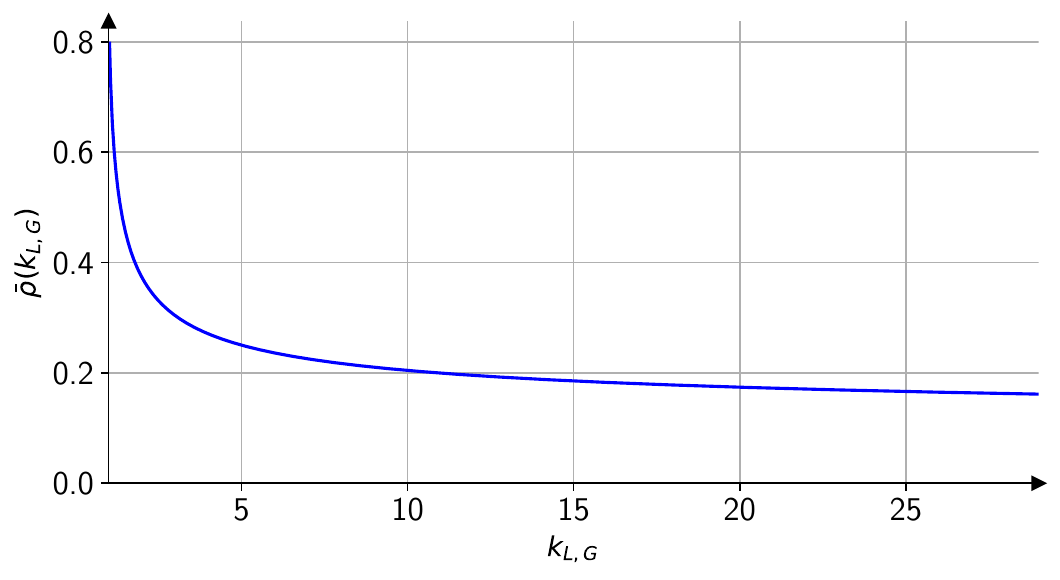}
\hspace{.02\linewidth}
\caption{Plot of the optimal contraction factor $\bar \rho(\prodeqnorm_{\loc,\glo})$ given by equation~\eqref{eq:optimal_contraction_factor}.}
\label{fig:optimal_rho}
\end{SCfigure}

%We are now ready to give the proof of Theorem~\ref{thm:GLin-ExpS_of_pfrlcn}, which follows from the results presented above. 

\subsection{Proof of Theorem~\ref{thm:GLin-ExpS_of_pfrlcn}}\label{sec:proof_theorem}
%The proof follows from the results of this Section.
We are now ready to give the proof of Theorem~\ref{thm:GLin-ExpS_of_pfrlcn}, which follows from the results of this Section. Indeed, given the assumptions of Theorem~\ref{thm:GLin-ExpS_of_pfrlcn}: (i) Theorem~\ref{thm:weak-contractivity_pfrlcn} implies that the \pfcn is weakly infinitesimally contracting on $\R^n$ with respect to $\norm{\cdot}_{2, \qw}$; (ii) Theorem~\ref{thm:loc_exp_stab_loc_contractivity} implies that the \pfcn is $c_{\eps}$-strongly infinitesimally contracting on $\ball{\qs_{\eps}}{\radiusstrong}$ with respect to $\norm{\cdot}_{2,\strongweight}$. Hence, Theorem~\ref{thm:GLin-ExpS_of_pfrlcn} follows from Corollary~\ref{cor:lin-exp-decay_GW-LS_contracting_systems} with $\mcS = \ball{\qs_{\eps}}{\radiusstrong}$, $\norm{\cdot}_{\glo} = \norm{\cdot}_{2, \qw}$ and $\norm{\cdot}_{\loc} = \norm{\cdot}_{2,\strongweight}$.

\section{Conclusions}  
In this paper, we proposed and analyzed two families of continuous-time firing-rate neural networks: the firing-rate competitive network, \fcn, and the positive firing-rate competitive network, \pfcn, to tackle sparse reconstruction and positive sparse reconstruction problems, respectively.
These networks arise from a top/down normative framework that aims to provide a biologically-plausible explanation for how neural circuits solve sparse reconstruction and other composite optimization problems.
This framework is based upon the theory of proximal operators for composite optimization and leads to continuous-time firing-rate neural networks that are therefore interpretable.

We first introduced a result relating the optimal solutions of the SR and positive SR problems to the equilibria of the \fcn and \pfcn (Lemma~\ref{lem:opt_sol-eq_point} and Corollary~\ref{cor:pfcn_equilibria}). Crucial for the \pfcn is the fact that this is a positive system (see Lemma~\ref{lem:positive_frcn}). This, in turn, can be useful to effectively model both excitatory and inhibitory synaptic connections in a biologically plausible way. Then, we investigated the convergence properties of the proposed networks: we provided an explicit convergence analysis for the \pfcn and gave rigorous conditions to extend the analysis to the \fcn.  Specifically, we showed that (i) the \pfcn~\eqref{eq:pfr-lca_x_dot} is weakly contracting on $\R^n$ (Theorem~\ref{thm:weak-contractivity_pfrlcn}); (ii) if the dictionary $\Phi$ is RIP, then the equilibrium point of the \pfcn is locally exponentially stable and, in a suitably defined norm, it is also strongly contracting in a neighborhood of the equilibrium (Theorem~\ref{thm:loc_exp_stab_loc_contractivity}). These results lead to Theorem~\ref{thm:GLin-ExpS_of_pfrlcn} that establishes linear-exponential convergence of the \pfcn.

To derive our key findings, we also devised a number of instrumental results, interesting {\em per se}, providing: (i) algebraic results on the log-norm of triangular matrices; (ii) convergence analysis for a broader class of non-linear dynamics (globally-weakly and locally-strongly contracting systems) that naturally arise from the study of the \fcn and \pfcn. Finally, we illustrated the effectiveness of our results via numerical experiments.

With our future research, we plan to extend our results to design networks able to tackle the sparse {\em coding} problem~\cite{CSNB-WG:16, POH:02, POH:03}, which involves {learning} features to reconstruct a given stimulus. We expect that tackling the sparse coding problem will lead to the study of RNNs with both neural and synaptic dynamics~\cite{DWD-JJH:92, LK-ML-JJES-EKM:20, VC-FB-GR:22g}. In this context, we plan to explore if Hebbian rules~\cite{DOH:49, WG-WK:02, VC-FB-GR:22g} can be effectively used to learn the dictionary. Moreover, it would be interesting to tackle SR problems with more general and non-convex sparsity-inducing cost functions~\cite{VC-SMF-DR:23}.
Finally, given the relevance and wide-ranging applications of globally-weakly and locally-strongly contracting systems, we will explore if tighter linear-exponential convergence bounds can be devised.

\section*{Acknowledgement}
The authors wish to thank Eduardo Sontag for stimulating conversations about contraction theory.

This work was in part supported by AFOSR project FA9550-21-1-0203 and NSF Graduate Research Fellowship under Grant No. 2139319. Giovanni Russo wishes to acknowledge financial support by the European Union - Next Generation EU, under PRIN 2022 PNRR, Project “Control of Smart Microbial Communities for Wastewater Treatment”.

\appendices
\section{Weight Matrix $\qw$ in Lemma~\ref{thm:contractivity_fnn} and Theorem~\ref{thm:weak-contractivity_pfrlcn}} 
\label{apx:weight_matrix_D}
We start with giving the explicit expression of the matrix $D \in \R^{n\times n}$ in Lemma~\ref{thm:contractivity_fnn}~\cite{VC-AG-AD-GR-FB:23c}.
To this purpose, we first recall that for any symmetric matrix $W \in \R^{n \times n}$, it is always possible to decompose $W$ into the form $W = U\Lambda U^\top$, where $U\in \R^{n\times n}$ is the orthogonal matrix whose columns are the eigenvectors of $W$, and $\Lambda = [\lambda] \in \R^{n\times n}$ is diagonal with $\lambda \in \R^n$ being the vector of the eigenvalues of $W$. Next, to define the weight matrix $\qw$, we need to introduce the function $\map{\fsplit}{]{-}\infty,1]}{[2,+\infty[}$ defined by
$
\fsplit(z):=2\big(1+\sqrt{1-z}\big), \forall z \in {]{-}\infty,1]}.
$
Then, letting $\fsplit(\Lambda):=\diag{(\fsplit(\lambda_1),\dots,\fsplit(\lambda_n))}$, it is
\beq \label{eq:weight_matrix_wc}
\qw := U \fsplit(\Lambda) U^\top\succ 0.
\eeq
The expression of the matrix $\qw$ in Theorem~\ref{thm:weak-contractivity_pfrlcn} follows from~\eqref{eq:weight_matrix_wc} when $W= \bigl(I_n-\Phi^\top \Phi\bigr)$. 

\section{On the Positiveness of the \pfcn}
In this appendix, we give a formal proof of the fact that the \pfcn~\eqref{eq:pfr-lca_x_dot} is a positive system. That is, the state variables are never negative, given a non-negative initial state. In order words, the positive orthant $\R_{\geq 0}^n$ is forward invariant. 
First, we recall the following standard:
\smallskip
\bd[Forward invariant set]
A set $\mcS \subset \R^n$ is \emph{forward invariant} with respect to the system~\eqref{eq:dynamical_system} if for every $x(0) \in \mcS$ it holds $\odeflowtx{t}{x(0)} \in \mcS$, for all $t \geq 0$.
\ed
\smallskip
Then, we give the following:
\smallskip
\begin{lem}[On the positiveness of the \pfcn]
\label{lem:positive_frcn}
The \pfcn~\eqref{eq:pfr-lca_x_dot} is a positive system.
\end{lem}
\begin{proof}
To prove the statement we prove that the positive orthant $\R_{\geq 0}^n$ is forward invariant for $\fpfcn$.
We recall that, by applying Nagumo's Theorem~\cite{MN:1942}, the positive orthant is forward invariant for a vector field $f$ if and only if
\beq
\label{eq:nagumo_pos}
f_i(x) \geq 0 \quad \forall x \in \R_{\geq 0}^n \text{ such that } x_i =0.
\eeq
Now, let us consider the \pfcn written in components
\[
\dot{x_i} = -x_i + \relu{\Bigl(-\sum_{j = 1, j \neq i}^n \trasp{\Phi_i}\Phi_j x_j(t) + \trasp{\Phi_i} u(t) - \lambda \Bigr)} = \subscr{f}{PFCN,i}(x), \quad i \in \until{n}.
\]
Then, for all $x \in \R_{\geq 0}^n$ such that $x_i =0$ we have
\[
\subscr{f}{PFCN,i}(x) = \relu{\Bigl(-\sum_{j = 1, j \neq i}^n \trasp{\Phi_i}\Phi_j x_j(t) + \trasp{\Phi_i} u(t) - \lambda \Bigr)} \geq 0,
\]
for each $i$. This concludes the proof.
\end{proof}
\label{app:positive_system}
\section{A Primer on Proximal Operators}
\label{apx:proximal_operator}
In this appendix, we provide a brief overview of proximal operators and outline the main properties needed for our analysis.
We start by giving a number of preliminary notions.

Given $\map{g}{\R^n}{\realextended}:= [{-}\infty,+\infty]$, the \emph{epigraph} of $g$ is the set $\operatorname{epi}(g) = \setdef{(x,y) \in \R^{n+1}}{g(x) \leq y}$.
\smallskip
\bd[Convex, proper, and closed function]
A function $\map{g}{\R^n}{\realextended}$ is
\bi
\item \emph{convex} if $\operatorname{epi}(g)$ is a convex set;
\item \emph{proper} if its value is never $-\infty$ and there exists at least one $x \in \R^n$ such that $g(x) < \infty$;
\item \emph{closed} if it is proper and $\operatorname{epi}(g)$ is a closed set.
\ei
\ed
Next, we define the proximal operator of $g$, which is a map that takes a vector $x \in \R^n$ and maps it into a subset of $\R^n$, which can be either empty, contain a single element, or be a set with multiple vectors.
\bd[Proximal Operator]
\label{apx:def:prox_operator}
The \emph{proximal operator} of a function $\map{g}{\R^n}{\realextended}$ with parameter $\gamma>0$, $\map{\prox{\gamma g}}{\R^n}{\R^n}$, is the operator given by
\begin{equation}
\prox{\gamma g}(x) = \argmin_{z \in \R^n} g(z) + \frac{1}{2 \gamma}\|x - z\|_2^2, \quad \forall x \in \R^n.
\end{equation}
\ed

Of particular interest for our analysis is the case when the $\prox{\gamma g}$ is a singleton. The next Theorem~\cite[Theorem 6.3]{AB:17} provides conditions under which the $\prox{\gamma g}$ exists and is unique.
\bt[Existence and uniqueness]
\label{apx:thm:ccp_uniq_prox_operator}
Let $\map{g}{\R^n}{\realextended}$ be a convex, closed, and proper function. Then $\prox{\gamma g}(x)$ is a singleton for all $x \in \R^n$.
\et

The above result shows that for a convex, closed, and proper function $g$, the proximal operator $\prox{\gamma g}(x)$ exists and is unique for all $x \in \R^n$.

Next, we recall a result on the calculus of proximal mappings~\cite[Section 6.3]{AB:17}.
\begin{lem}[Prox of separable functions]
\label{apx:lem:prox_of_separable_functions}
Let $\map{g}{\R^n}{\R}$ be a convex, closed, proper, and separable function, that is $g(x) = \sum_{i=1}^n g_i(x_i)$, with $\map{g_i}{\R}{\R}$ being convex, closed, and proper functions. Then 
$$(\prox{\gamma g}(x))_i = \prox{\gamma g_i}(x_i), \quad i \in \until{n}.$$
\end{lem}

\subsection{Proximal Gradient Method}
\label{apx_proximal_gradient_method}
Based on the use of proximal operators, {\em proximal gradient method} (see, e.g.,~\cite{NP-SB:14}) can be devised to iteratively solve a class of composite (possibly non-smooth) convex problems
\beq
\label{apx:eq:composite_problem}
\min_{x \in \R^n} f(x) + g(x),
\eeq
where $\map{f}{\R^n}{\R}$, $\map{g}{\R^n}{\realextended}$ are convex, proper and closed functions, and $f$ is differentiable. At its core, the proximal gradient method updates the estimate of the solution of the optimization problem by computing the proximal operator of $\alpha g$, where $\alpha > 0$ is a step size, evaluated at the difference between the current estimate and the gradient of $\alpha f$ computed at the current estimate. That is,
\[
x^{k+1} := \prox{\alpha^k g}{\bigl(x^k - \alpha^k \nabla f(x^k)\bigr)}.
\]

Notably, this method has been recently extended and generalized to a continuous-time framework~\cite{SHM-MRJ:21, AD-VC-AG-GR-FB:23f}, resulting in solving a continuous-time FNN. In this case, the iteration becomes the \emph{continuous-time proximal gradient dynamics}
\beq
\label{apx:eq:prox:gradient}
\dot x = - x + \prox{\gamma g}{\bigl(x - \gamma \nabla f(x)\bigr)},
\eeq
with $\gamma >0$.

Finally, we note that for $f(x) := \frac{1}{2}\big\|u- \Phi y\big\|^2_2$ and $g(x) := \lambda S(y)$ the composite optimization problem~\eqref{apx:eq:composite_problem} is the SR problem~\eqref{eq:sparse_approx_C}. Moreover, we get the \fcn~\eqref{eq:frlca-x_dot_general} by setting $\gamma = 1$ in~\eqref{apx:eq:prox:gradient}.

\subsection{Proximal Operator for $\lambda S$ in the Positive SR Problem}
We provide the explicit computation of the proximal operator of the sparsity-inducing term of the positive SR problem~\eqref{eq:positive_E_lasso_unconstrained}.

\begin{lem}
\label{apx:lem:prox_of_relu}
Consider $\lambda >0$ and let $\map{S_1}{\R^n}{\R}$, $S_1(y) = \|y\|_1 + \frac{1}{\lambda}\iota_{\R^n_{\geq0}}(y)$, $\forall y \in \R^n$. Then
$$\prox{\lambda S_1}(y) = \relu(y - \lambda \1_n).$$
\end{lem}
\begin{proof}
We start by noticing that $\lambda S_1$ is separable across indices and, for any $y_i \in \R$, we have $\lambda s_1(y_i) = \lambda y_i + \iota_{\R_{\geq0}}(y_i)$. Hence, Lemma~\ref{apx:lem:prox_of_separable_functions} implies that
the computation of the proximal operator of $\lambda S_1$ reduces to computing scalar proximals of $\lambda s_1(y_i)$. This can be done as follows:
\begin{align*}
\prox{\lambda s_1}(y_i) &= \argmin_{z \in \R} \frac{1}{2}(y_i - z)^2 + \lambda z + \iota_{\R_{\geq0}}(z) =
\begin{cases}
0 & \textup{ if } y_i \leq \lambda, \\
y_i - \lambda & \textup{ if } y_i > \lambda.
\end{cases}
\end{align*}
Note that, by definition of (shifted) $\relu$ function this is exactly $\relu(y_i - \lambda)$. In turn, this proves the statement.
\end{proof}

\section{The $\ell_2$ Logarithmic Norm of Upper Triangular Block Matrices}
\label{axp:l2_log_norm_of_upper_triangular_block_matrices}
We present an algebraic result of the $\ell_2$ log-norm of upper triangular block matrices. This result is useful for determining the rate and norm with respect to which the \pfcn exhibits strong infinitesimal contractivity, as stated in Theorem~\ref{thm:loc_exp_stab_loc_contractivity}. The following Lemma is inspired by~\cite[E2.28]{FB:23-CTDS}. We also refer to~\cite{GR-MDB-EDS:10a} for a result on the log-norm of these triangular matrices using non-Euclidean norms.

\smallskip

\begin{lem}[The $\ell_2$ logarithmic norm of upper triangular block matrices]
\label{lem:mu2-for-upper-triang-matrix}
Consider the block matrix
\begin{equation*}
A =
\begin{bmatrix}
A_{11} & A_{12} \\
0 & A_{22}
\end{bmatrix} \; \in\R^{(n+m)\times(n+m)}.
\end{equation*}
For all $\eps>0$ and for $P_\eps =
\begin{bmatrix}
\eps P_1 & 0 \\
0 & \eps^{-1} P_2
\end{bmatrix}$ with $P_1=P_1^\top \succ 0$ and $P_2=P_2^\top \succ 0$, we have
\begin{equation}
\lognorm{A}{2,P_\eps^{1/2}} \leq
\max\big\{ \lognorm{A_{11}}{2,P_1^{1/2}}, \lognorm{A_{22}}{2,P_2^{1/2}}\big\} + \eps
\norm{P_1^{1/2} A_{12} P_2^{-1/2}}_{2}.
\end{equation}
\end{lem}
\begin{proof}
We compute
\begin{align*}
\lognorm{A}{2,P_\eps^{1/2}} 
&=\Bigglognorm{ \begin{bmatrix} P_1^{1/2} A_{11} P_1^{-1/2} & \eps P_1^{1/2} A_{12} P_2^{-1/2}\\
0 & P_2^{1/2} A_{22}P_2^{-1/2} \end{bmatrix}}{2} \nonumber\\
&= \Bigglognorm{ \begin{bmatrix} P_1^{1/2} A_{11} P_1^{-1/2} & 0  \\ 0 & P_2^{1/2} A_{22} P_2^{-1/2} \end{bmatrix} + \begin{bmatrix} 0 & \eps P_1^{1/2} A_{12} P_2^{-1/2} \\ 0 & 0 \end{bmatrix}}{2} \\
&\leq \Bigglognorm{ \begin{bmatrix} P_1^{1/2} A_{11} P_1^{-1/2} & 0 \\ 0 & P_2^{1/2} A_{22} P_2^{-1/2} \end{bmatrix}}{2} + \eps \Biggnorm{\begin{bmatrix} 0 & P_1^{1/2} A_{12} P_2^{-1/2} \\ 0 & 0 \end{bmatrix}}_{2},
\end{align*}
where the last inequality follows by applying the translation property of the log-norm and the inequality $\mu(B) \leq \norm{B}$, for all matrix $B$. From the LMI characterization of the $\ell_2$ logarithmic norm, we obtain
\begin{equation*}
\Bigglognorm{
\begin{bmatrix} P_1^{1/2} A_{11}P_1^{-1/2}  &  0 \\ 0 & P_2^{1/2} A_{22}P_2^{-1/2} \end{bmatrix} }{2} =\max\left\{ \lognorm{A_{11}}{2,P_1^{1/2}}, \lognorm{A_{22}}{2,P_2^{1/2}}  \right\}.
\end{equation*}
The claim then follows by noting that
$\Biggnorm{\begin{bmatrix} 0 & P_1^{1/2} A_{12} P_2^{-1/2} \\ 0 & 0 \end{bmatrix}}_{2} = \norm {P_1^{1/2} A_{12} P_2^{-1/2}}_{2}$.
\end{proof}

Next, we give a specific result for a particular case of the matrix $A$ (which has the same form as the Jacobian of \pfcn computed at the equilibrium). In this particular case, we are able to determine and specify the matrices $P_1$ and $P_2$.

\smallskip

\begin{cor}
\label{cor:mu2_bound-for-Df}
Consider the block matrix
\begin{equation*}
B =
\begin{bmatrix}
- B_{11} & B_{12} \\
0 & - I_{m}
\end{bmatrix} \; \in\R^{(n+m)\times(n+m)},
\end{equation*}
with $B_{11} = B_{11}^\top \succ 0$ satisfying $\subscr{\lambda}{min}(B_{11}) \geq l$, with 
$l \in {]0,1]}$. Then, for all $\eps > 0$ and for $Q_\eps =
\begin{bmatrix}
\eps I_{n} & 0 \\
0 & \eps^{-1} I_{m}
\end{bmatrix}$, we have
\begin{equation}
\label{eq:lognorm(D)}
\lognorm{B}{2,Q_\eps} \leq - \bigl(l - \eps^2 \norm{B_{12}}_{2}\bigr).
\end{equation}
\end{cor}
\begin{proof}
By applying Lemma~\ref{lem:mu2-for-upper-triang-matrix} to the block matrix $B$, for all $\eps > 0$ we have
\begin{align*}
\lognorm{B}{2,Q_\eps} & \leq
\max\big\{ \lognorm{-B_{11}}{2,I_{n}}, \lognorm{- I_{m}}{2,I_{m}}\big\} + \eps^2
\norm{I_{n} B_{12} I_{m}}_{2}\\
&\leq \max\big\{ - l , - 1 \big\} +  \eps^2\norm{B_{12}}_{2} = - l  + \eps^2\norm{B_{12}}_{2}.
\end{align*}
This concludes the proof.
\end{proof}

\smallskip
\begin{rem}
The result in Corollary~\ref{cor:mu2_bound-for-Df} implies that:
\begin{enumerate}
\item if $\norm{B_{12}}_{2} = 0$, then $\lognorm{B}{2,Q_\eps} < 0$ for all $\eps > 0$;
\item if $\norm{B_{12}}_{2} \neq 0$, then $\lognorm{B}{2,Q_\eps} < 0$, for all $\ds \eps \in {]0, \sqrt{l/\norm{B_{12}}_{2}}[}$.
\end{enumerate}
\end{rem}

\smallskip

\subsubsection{Expression of $\qs_\eps$ and $\ce$ in Theorem~\ref{thm:loc_exp_stab_loc_contractivity}}
Corollary~\ref{cor:mu2_bound-for-Df} enables the computation of the rate and the norm with respect to which the \pfcn is strongly infinitesimally contracting, as stated in Theorem~\ref{thm:loc_exp_stab_loc_contractivity}. In fact, the Jacobian of $\fpfcn$ computed at the equilibrium, given by equation~\eqref{eq:DflcaPFR(xstar)}, is in the form of the matrix $B$ in Corollary~\ref{cor:mu2_bound-for-Df} with $n = \na$, $m = \nia$, $B_{11} := \subscr{\trasp{\Phi}}{a}\subscr{\Phi}{a} \succ 0$, $ l = 1 - \delta \in {]0,1]}$, and $B_{12} := -\subscr{\trasp{\Phi}}{a}\subscr{\Phi}{ia}$.
Therefore the \pfcn is strongly infinitesimally contracting w.r.t. the norm $\norm{\cdot}_{2,S_\eps}$ with rate $\ce$, where
\beq
\label{eq:c_exp}
\begin{aligned}
\text{(i) } \qs_\eps := Q_\eps \text{ and } \ce &= 1 - \delta, \forall \eps > 0, &\text{ if } \norm{\subscr{\trasp{\Phi}}{a}\subscr{\Phi}{ia}}_{2} = 0;\\
\text{(ii) } \qs_\eps := Q_\eps \text{ and } \ce &= 1 - \delta - \eps^2 \norm{\subscr{\trasp{\Phi}}{a}\subscr{\Phi}{ia}}_{2}, \forall \ds \eps \in {]0, \sqrt{(1 - \delta)/\norm{\subscr{\trasp{\Phi}}{a}\subscr{\Phi}{ia}}_{2}}[}, &\text{ if } \norm{\subscr{\trasp{\Phi}}{a}\subscr{\Phi}{ia}}_{2} \neq 0.
\end{aligned}
\eeq

\bibliographystyle{plainurl+isbn}
\bibliography{alias, Main, FB, Veronica}
\end{document}

%% file: my_commands.tex
\newtheorem{definition}{Definition}
\newtheorem{thm}{Theorem}[section]
\newtheorem{lem}[thm]{Lemma}
\newtheorem{pro}[thm]{Proposition}
\newtheorem{cor}[thm]{Corollary}
\newtheorem{remark}{Remark}
\newtheorem{assumption}{Assumption}
\newenvironment{rem}{\begin{remark}}{\qed\end{remark}}

\newtheorem*{metatheorem}{Informal statement}
\newcommand{\bd}{\begin{definition}} 
\newcommand{\ed}{\end{definition}} 
\newcommand{\bp}{\begin{pro}} 
\newcommand{\ep}{\end{pro}} 
\newcommand{\bt}{\begin{thm}} 
\newcommand{\et}{\end{thm}}
\newcommand{\blm}{\begin{lem}} 
\newcommand{\elm}{\end{lem}}
\let\eps\varepsilon

\newcommand{\Z}{\mathbb{Z}}
\newcommand{\R}{\mathbb{R}}

\newcommand{\mcS}{\mathcal{S}}

\newcommand{\OF}{\mathsf{F}}
\newcommand{\OG}{\mathsf{G}}

\DeclareMathOperator{\rank}{\operatorname{rank}}

\newcommand{\realextended}{\overline{\R}}
\newcommand{\bi}{\begin{itemize}} 
\newcommand{\ei}{\end{itemize}} 
\newcommand{\bds}{\begin{description}} 
\newcommand{\eds}{\end{description}} 
\newcommand{\beq}{\begin{equation}} 
\newcommand{\eeq}{\end{equation}} 
\newcommand{\abs}[1]{\left|#1\right|}

\newcommand{\prox}[1]{\mathrm{prox}_{#1}}
\newcommand{\norm}[1]{\|#1\|}
\newcommand{\normtwo}[1]{\|#1\|_2}

\newcommand{\derp}[2]{\frac{\partial #1}{\partial #2}}

\newcommand{\trasp}[1]{#1^{\textsf{T}}}
\newcommand{\diag}[1]{[#1]}

\newcommand{\Biggnorm}[1]{\Bigg\|#1\Bigg\|}

\newcommand{\lognorm}[2]{\mu_{#1}(#2)}
\renewcommand{\lognorm}[2]{\mu_{#2}(#1)} %22/06/2023
\newcommand{\wlognorm}[3]{\mu_{#1,#2}\bigl(#3\bigr)}
\newcommand{\Bigglognorm}[2]{\mu_{#2}\Biggl(#1\Biggr)}
\newcommand{\until}[1]{\{1,\dots, #1\}}

\newcommand{\subscr}[2]{#1_{\textup{#2}}}
\newcommand{\supscr}[2]{#1^{\textup{#2}}}
\newcommand{\setdef}[2]{\{#1 \; | \; #2\}}

\newcommand{\map}[3]{#1\colon #2 \rightarrow #3}

\DeclareMathOperator*{\argmin}{argmin}

\newcommand{\ceil}[1]{\lceil #1 \rceil}
\newcommand{\Bigceil}[1]{\left\lceil #1 \right\rceil}
\newcommand{\ds}{\displaystyle}

\newcommand{\xstar}{x^{\star}}
\newcommand{\e}{\mathrm{e}}
\newcommand{\FNN}{FNN\xspace}
\newcommand{\lmax}{\alpha(W)}
\newcommand{\hop}{\subscr{x}{H}}
\newcommand{\fr}{\subscr{x}{F}}
\newcommand{\hopdot}{\subscr{\dot x}{H}}
\newcommand{\frdot}{\subscr{\dot x}{F}}
\newcommand{\uh}{\subscr{u}{H}}
\newcommand{\ufr}{\subscr{u}{F}}

\newcommand{\Q}{Q}
\newcommand{\fsplit}{\theta}

\newcommand{\na}{\subscr{n}{a}}
\newcommand{\nia}{\subscr{n}{ia}}
\newcommand{\fcn}{FCN\xspace}
\newcommand{\pfcn}{PFCN\xspace}
\newcommand{\fpfcn}{\subscr{f}{\pfcn}}
\newcommand{\Elasso}{\subscr{E}{L}}
\newcommand{\nondiffset}[1]{\Omega_{#1}}

\newcommand{\odeflowtx}[2]{\phi_{#1}\bigl({#2}\bigr)}
\newcommand{\xtmp}{\subscr{x}{tmp}}
\newcommand{\qw}{D}%{\subscr{Q}{W}}
\newcommand{\qs}{S}%{\subscr{Q}{S}}
\newcommand{\radius}{r}
\newcommand{\radiusstrong}{q}

\newcommand{\tld}{\subscr{t}{cross}}

\newcommand{\eqnorm}{k}
\newcommand{\prodeqnorm}{k}
\newcommand{\ce}{\subscr{c}{exp}}
\newcommand{\cl}{\subscr{c}{lin}}
\newcommand{\loc}{\textup{L}}
\newcommand{\glo}{\textup{G}}

\newcommand{\kab}{\eqnorm_\alpha^\beta}
\newcommand{\kba}{\eqnorm_\beta^\alpha}

\newcommand{\ratiolg}{\prodeqnorm_{\loc,\glo}}

\newcommand{\glbnorm}[1]{\|#1\|_{\glo}}
\newcommand{\ball}[2]{B^{#1}_{#2}}
\newcommand{\softt}[2]{\operatorname{soft}_{#1}({#2})}
\newcommand{\soft}[1]{\operatorname{soft}_{#1}}
\newcommand{\sign}[1]{\operatorname{sign}({#1})}
\newcommand{\relu}{\operatorname{ReLU}}

\newcommand{\1}{\mbox{\fontencoding{U}\fontfamily{bbold}\selectfont1}}
\newcommand{\0}{\mbox{\fontencoding{U}\fontfamily{bbold}\selectfont0}}